\documentclass[letterpaper, 10 pt, conference]{ieeeconf}  

\newcommand{\B}{\color{black}}

\newtheorem{theorem}{Theorem}
\newtheorem{lem}[theorem]{Lemma}

\IEEEoverridecommandlockouts                              
\usepackage{graphics} 
\usepackage{epsfig} 
\usepackage{mathptmx} 
\usepackage{times} 
\usepackage{amsmath} 
\usepackage{amssymb}  

\usepackage{subfig}
\usepackage{color}

\graphicspath{{figures/}}
\title{\LARGE \bf
Different Environment Feedback in Fast-slow Eco-evolutionary Dynamics
}

\author{Lulu Gong, and Ming Cao 
\thanks{ L. Gong and M. Cao are with ENTEG,  Faculty of Science and Engineering,
        University of Groningen, 9747 AG, Groningen, The Netherlands.
        {\tt\small l.gong@rug.nl, m.cao@rug.nl}. The work of L. Gong was supported in part by China scholarship council (CSC). The work of M. Cao was supported in part by the European Research Council (ERC-CoG-771687) and the Netherlands Organization for Scientific Research (NWO-vidi-14134).}%
}

\begin{document}

\maketitle
\thispagestyle{empty}
\pagestyle{plain}

\begin{abstract}
 The fast-slow dynamics of an eco-evolutionary system are studied, where we consider the feedback actions of  environmental resources that are classified into those that are self-renewing and those externally supplied. We show although these two types of resources are drastically different, the resulting closed-loop systems bear close resemblances, which include the same equilibria and their stability conditions on the boundary of the phase space,  and the similar appearances of equilibria in the interior. After closer examination of specific choices of parameter values, we disclose that the global dynamical behaviors of the two types of closed-loop systems can be fundamentally different in terms of limit cycles: the system with self-renewing resources undergoes a generalized Hopf bifurcation such that one stable limit cycle and one unstable limit cycle can coexist; the system with externally supplied resources can only have the stable limit cycle induced by a supercritical Hopf bifurcation. {\B Finally, the explorative analysis is carried out to show the discovered dynamic behaviors are robust in even larger parameter space}.


\end{abstract}

\section{Introduction}
To study complex biological and ecological systems, a range of mathematical models have been proposed, among which the Lotka–Volterra model \cite{lotka1956elements}, \cite{volterra1926variazioni} and MacArthur's consumer-resource model \cite{Macarthur1967}, \cite{CHESSON199026} have been playing prominent roles. The Lotka–Volterra model essentially describes the evolution of the predator and prey populations. The consumer-resource model  covers a wider range of ecosystems, by taking into account diverse resource types. For another class of models, evolutionary game theory has served as the basis driven by the fitness consideration.
 Evolutionary game dynamics study the strategic evolution of one or multiple populations of individuals \cite{Hofbauer:98}, where in  the  classic setting,  the  payoffs  in  each pairwise-interaction game are usually predetermined and specified by constant payoff matrices. Recently, an
eco-evolutionary game model,  introduced by  evolutionary biologists in a series of seminal works \cite{Weitz2016}, \cite{Lin2019}, \cite{Tilman:20}, provides the description for the ecosystems from a dynamic-feedback game theoretical point of view. This model explicitly incorporates both games and resource dynamics into an integrated system,  and  the resulting coupled system dynamics are investigated accordingly.

The eco-evolutionary  model
is first proposed in \cite{Weitz2016}, and is then  generalized in \cite{Tilman:20} to consider different resource dynamics. In \cite{Tilman:20}, it has been shown that the eco-evolutionary systems with renewable and decaying resources have similar behaviors in that their dynamics are qualitatively analogous.  By varying the time-scale parameters, limit cycle behaviors are observed in both systems. Note that in the renewable resource case the population's actions are considered to consume the resource, while the population's actions will produce the resource in the decaying  resource case. 

In the classic consumer-resource model, the resources are usually considered to be consumed by a population of consumers. And the influence of different types of resources on the overall system properties is the usual subject of study \cite{Cui2020}. Different from the existing eco-evolutionary models, in this work, we consider that  all the strategic players in the population  consume or harvest the resource with different harvest efforts. And two distinct resource models, self-renewing and externally supplied resources, are studied. { \B Although our first model with the self-renewing resource is similar to the renewable resource model studied in \cite{Tilman:20}, our second model is drastically different from the decaying resource model in \cite{Tilman:20} in that the strategists are still considered to be consumers in our model, which is consistent with the classic setting in consumer-resource models.}

By extending the concept of resources,  we study  the eco-evolutionary dynamics under the different resource models using two-player two-strategy games. First, we show that  these two systems bare some similarities in terms of boundary equilibria and their stability conditions. The two systems  can have interior equilibria which can be continuous and thus infinite, and can also be two isolated points. Then, to further analyze the stability of interior equilibria, we fix some parameters and conduct bifurcation analysis for one payoff parameter and the time-scale parameter. 
Through both one-parameter and two-parameter analyses, it is shown that both systems can exhibit Hopf bifurcations in the specialised parameter conditions. However, the system with the self-renewing resource can further have the generalized Hopf bifurcation, which is impossible for the other system. Therefore, in this situation it is natural for the first system to have the coexistence of two limit cycles, one stable and the other unstable, but the second system can only 
possess one stable limit cycle. 

We emphasize that although limit cycles have been identified in relevant eco-evolutionary research \cite{Tilman:20}, \cite{rand2017}, \cite{bever1997}, \cite{henrich2004},  \cite{Gong:2020}, unstable and double limit cycles have not been reported in the relevant studies so far. These behaviors correspond to an interesting dynamic bistable phenomenon which has been discovered in Lotka-Volterra models \cite{bistability}, \cite{Yu}. More importantly, we also implement analysis revealing how the two limit cycles can coexist through two-parameter bifurcation. It has been clarified that the time-scale difference and the payoff variation can have great influence on the eco-evolutionary dynamics, and as a result the two systems' behaviors become different and complicated. {\B Additional analysis shows that the discovered dynamic behaviors persist in other parameter conditions.}

The rest of the paper is organized as follows. Section II
introduces the two eco-evolutionary systems. Section III shows some preliminary results in the general settings. Section IV further provides the detailed bifurcation and stability analyses for the interior equilibria, and some numeric examples are given to illustrate the theoretical results. {\B Furthermore, brief analysis on the robustness of the discovered results is given in Section V.} Conclusions are
drawn in Section VI.

\section{Problem setup}
\subsection{Fast-slow eco-evolutionary systems}
Consider in an infinite well-mixed population, the individuals play two-player matrix games with  the strategy set  $\mathcal{S}=\{1,2,...,n\}$. The payoffs for the players are affected by an external environmental factor. {\B We use replicator equations to describe the evolution of strategies in the population, and assume the resource dynamics evolve more slowly than the strategic evolution, which is natural in the ecological applications \cite{Tilman:20}, \cite{hastings2018transient}.  Combining with the resource dynamics,
we obtain the so-called} \emph{eco-evolutionary system}
\begin{equation}\label{EEmodel}
\begin{aligned}
&\epsilon\dot{x}_i=x_i[(A(m)x)_i-x^T A(m)x],\\
&\dot{m}=f(m)-h(x,m),
\end{aligned}
\end{equation}
where $x_i \in [0,1]$ and $m \in \mathbb{R}_+$ respectively denote the frequency of strategy $i$ and the amount of an environmental resource of interest; $x=[x_1,...,x_n]^T$;  $A(m)$ is the environment-dependent payoff matrix; $f(m)$ represents the intrinsic dynamics of the resource, and $h(x,m)$ captures the influence of the population's actions on the resource; $0<\epsilon  \ll 1$ is a small parameter accounting for the difference of time scales between the variables of strategies and resource.

\subsection{Different Resource Dynamics}
In nature, ecosystems can involve many types of resources, among which two typical and contrasting examples are biotic and abotic resources.  The biotic resource, such as forests and animals, usually can self-replicate and thus grows logistically if the influence of the population is absent. Thus its intrinsic growth can be represented by a logistic function \cite{Macarthur1967}, \cite{Weitz2016}, \cite{Tilman:20},
\begin{equation}
    f_1(m)=m(1-m/\kappa),
\end{equation}
where $\kappa$ is the 
carrying capacity of the resource.

In contrast, the abiotic resource, such as
minerals and small molecules, cannot self-replicate and is
usually supplied externally to an ecosystem \cite{Cui2020}. And its intrinsic dynamics can be captured by a liner function \cite{tikhonov2016}, \cite{Cui2020}
\begin{equation}
    f_2(m)=K-w m,
\end{equation}
where $K$ is the constant flux from some external resource supplier, and $w$ is the  self-decay rate of the resource. 

The difference of these two functions is that $f_1(m)$ is concave while $f_2(m)$ is strictly decreasing, as shown in Fig. \ref{fig:twofunctions}.
\begin{figure}[htbp!]
   \centering
   	\includegraphics[width=6cm]{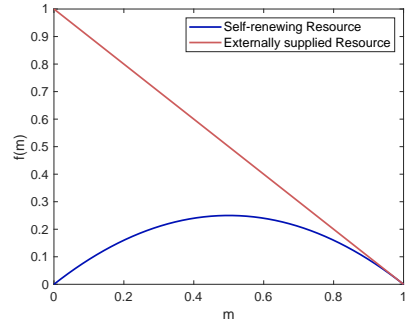}
    \caption{Two different resource models with $\kappa=K=w=1$.}
    \label{fig:twofunctions}
\end{figure}

It is considered that all the strategists  consume or harvest the resource, but different strategies have distinct harvest efforts $e_i>0, e_i\neq e_j, \forall  i, j\in \mathcal{S}$. And the function $h(x,m)$ has the form
\[h(x,m)=qm\sum_{i\in \mathcal{S}} e_ix_i,\]
where parameter $q$  maps the harvest efforts  into the rate of reduction in the resource. 

When the intrinsic dynamics of the resource change, accordingly the feedback rules of the resource to the strategic dynamics are significantly different. In this work, we intend to investigate the  eco-evolutionary dynamics under these two different resource feedback rules.
\subsection{Normalized Models}
We consider the two-strategy game, {\B and assume the second strategy having higher harvest effort on the resource than the first.} Then we only need to focus on one strategy evolution because $x_1+x_2=1$. For the self-renewing resource case, we can use the change of variable 
$r=\frac{m-\kappa(1-qe_2)}{q\kappa(e_2-e_1)}$
to normalize the resource variable such that $r$ is confined in the unit interval $[0,1]$. {\B And we consider that the payoff matrix is linearly dependent on $r$ and is a convex combination of the payoff matrices in the two extreme cases, i.e., when $r=0$ and $1$. Following the convention \cite{Weitz2016}, \cite{Tilman:20}, \cite{Gong:2020}, then we have}
\begin{equation}\label{payoffmatrix}
A(r)= (1-r)\begin{bmatrix}
R_0&S_0\\
T_0&P_0
\end{bmatrix} +r\begin{bmatrix}
R_1&S_1\\
T_1&P_1
\end{bmatrix}.  
\end{equation}
{\B Substituting $A(r)$ into the replicator equation and after a re-scaling of the system time by using $\tau=t/\epsilon$,} we obtain the eco-evolutionary system   under self-renewing resource model as
\begin{equation}\label{eesr2}
\Sigma_s:\begin{cases}
\begin{aligned}
&\begin{aligned}x'=&x(1-x)((b-a+d-c)xr+(a-b)x\\
    &-(b+d)r+b),
\end{aligned}\\
&r'=\epsilon(1-q(e_1r+e_2(1-r)))(x-r),
\end{aligned}  
\end{cases}
\end{equation}
where $'$ represents the time derivative with respect to $\tau$.  For simplicity, we denote $a=R_0-T_0$, $b=S_0-P_0$, $c=T_1-R_1$, and $d=P_1-S_1$. {\B Note that each parameter of $a$, $b$, $c$, and $d$ now represents the payoff difference between the two strategies in the extreme cases.}

In a similar manner, for the externally supplied resource case, after the change of variable 
$r=\frac{(qe_1+w)(m(qe_2+w)-K)}{Kq(e_2-e_1)}$
and the time re-scaling, system (\ref{EEmodel}) under externally supplied resource model can be transformed into
\begin{equation}\label{eees2}
\Sigma_e:\begin{cases}
\begin{aligned}
&\begin{aligned}x'=&x(1-x)((b-a+d-c)xr+(a-b)x\\
    &-(b+d)r+b),
    \end{aligned}\\
&r'=\epsilon(x(1-r)(q e_1+w)-r(1-x)(q e_2+w)).
\end{aligned}
\end{cases}
\end{equation}
Note that in both  systems (\ref{eesr2}) and (\ref{eees2}), it is assumed that $0< e_1<e_2<1/q$ due to its ecological implication \cite{Tilman:20}.

\section{Preparatory analysis}
For systems (\ref{eesr2}) and (\ref{eees2}), the phase spaces are the unit square, i.e., $[0,1]^2$. There are $4$ edges, of which the two edges, $\{(x,r)\in[0,1]^2: x=0\}$ and $\{(x,r)\in[0,1]^2: x=1\}$, are invariant under the two systems. The systems have the same equlibria on the edges, i.e., $(0,0)$ and $(1,1)$.
{\B These two equilibria correspond to the situations of full (resp. none) high consuming strategists and depleted (resp. abundant) resource.}  These two equlibria's stability is readily to check by the Jacobians.  For system (\ref{eesr2}), the Jacobian at $(0,0)$ is 
\[J_s(0,0)=\begin{bmatrix}
b&0\\
\epsilon(1-qe_2)&-\epsilon(1-qe_2)
\end{bmatrix}.\]
Thus, when $b<0$, $(0,0)$ is asymptotically stable, and it will be unstable if $b>0$. 
The Jacobian at $(1,1)$ is 
\[J_s(1,1)=\begin{bmatrix}
c&0\\
\epsilon(1-qe_1)&-\epsilon(1-qe_1)
\end{bmatrix}.\]
So when $c<0$, $(1,1)$ is stable;  if $c>0$ it will be unstable.

Similarly, for system (\ref{eees2}),  the Jacobians of the two edge equilibria are
\[J_e(1,1)=\begin{bmatrix}
b&0\\
\epsilon(qe_1+w)&-\epsilon(qe_2+w)
\end{bmatrix},\]
\[J_s(0,0)=\begin{bmatrix}
c&0\\
\epsilon(qe_2+w)&-\epsilon(qe_1+w)
\end{bmatrix}.\]
It is noted that the stability conditions for these two equilibria are the same as the conditions for system (\ref{eesr2}).

Systems (\ref{eesr2}) and (\ref{eees2}) can also have some  equilibria in the interior of the domain, i.e., $(0,1)^2$. For the inteiror equilibra, we have the following observation.

\begin{lem}
Both systems (\ref{eesr2}) and (\ref{eees2})  can have at most $2$ isolated interior equilibria.
\end{lem}

\begin{proof}
Consider system (\ref{eesr2}).
Note that the term $(1-q(e_1r+e_2(1-r))$ is always positive for $r\in [0,1]$ because $0\leq e_1<e_2<1/q$.
The interior equilibria of system (\ref{eesr2}) are the solutions to
\begin{equation}\label{equilibria1}
\begin{aligned}
&(b-a+d-c)xr+(a-b)x-(b+d)r+b=0,\\
&x-r=0.
\end{aligned}    
\end{equation}
Substituting $x=r$ into the first equation of (\ref{equilibria1}) yields 
\begin{equation}\label{equilibria2}
 (b-a+d-c)r^2+(a-2b-d)r+b=0.   
\end{equation}
It can be checked that when $a=d, b=c=0$, system (\ref{eesr2}) has a continuum of equilibria: $\{(x,r)\in (0,1)^2: x=r\}$.
In other cases, (\ref{equilibria2}) will have at most two real solutions. In view of the fact that the interior equilibrium has to be in $(0,1)^2$, system (\ref{eesr2}) can have at most $2$ isolated interior equilibria. 

For system (\ref{eees2}), the interior equilibria are the solutions to
\begin{equation}\label{equilibria3}
\begin{aligned}
&(b-a+d-c)xr+(a-b)x-(b+d)r+b=0,\\
&x(1-r)(qe_1+w)-r(1-x)(qe_2+w)=0.
\end{aligned}    
\end{equation}
In the same calculation, after the substitution of $x=\frac{(qe_2+w)r}{(qe_1+w)+q(e_2-e_1)r}$, one obtains 
\begin{equation}\label{equilibria4}
 \begin{aligned}
 0=&[(b-a+d-c)(qe_2+w)-q(b+d)(e_2-e_1)]r^2\\
 &+[q(ae_2-2be_1-de_1)+w(a-2b-d)]r+b(qe_1+w).
 \end{aligned}
\end{equation}
Similarly, one can conclude that when all the coefficients of (\ref{equilibria4}) are $0$, system (\ref{eees2}) will have a set of equilibria $\{(x,r)\in (0,1)^2:x=\frac{(qe_2+w)r}{(qe_1+w)+q(e_2-e_1)r}\}$; otherwise, system (\ref{eees2}) has at most 2 isolated interior equilibria.
\end{proof}

In the following section, we carry out bifurcation analyses.

\section{Bifurcation and Stability Analyses}
We have identified that  systems (\ref{eesr2}) and (\ref{eees2}) may have $0$ to $2$ isolated interior equilibria. Due to the presence of too many parameters, it is difficult to analyze their stability collectively. Hence, in this section, we let some parameters be fixed specifically, and conduct bifurcation analysis for the interior equilibria with respect to the remaining parameters.  {\B Some payoff parameters, such as $a$, $b$, $d$, can be chosen to be positive and $c$ can vary along the real axis, so that we can focus on one payoff parameter and examine its impact on the system behaviors. To distinguish the two strategies, their harvest efforts should be vastly different. The ecological parameters $q$ and $w$ are not important in our study, and thus can be normalized to be $1$. It has been identified that in such eco-evolutionary systems the time-scale difference is important, so it is also of interest to investigate its influence in our models.} Therefore, we set $q=1$, $w=1$, $e_1=0.2$, $e_2=0.8$, and $a=0.2$, $b=0.1$, $d=0.4$. And we choose $c \in \mathbb{R}$ and $\epsilon \in (0,1]$ as the free parameters.

Under the specific setting, equations (\ref{equilibria2}) and (\ref{equilibria4}) become 
\begin{equation}\label{equilibria22}
 (0.3-c)r^2-0.4r+0.1=0, 
\end{equation}
\begin{equation}\label{equilibria42}
 (0.4-3c)r^2-0.6r+0.2=0. 
\end{equation}
Now, we will discuss the two systems separately.

\subsection{Analysis of system (\ref{eesr2})}
\subsubsection{Hopf bifurcation}
For system (\ref{eesr2}), denote the interior equilibria  by $(x_s^*, r^*_s)$. As $x_s^*=r^*_s$ in this case, it is easy to determine the location of the interior equilibria by solving equation (\ref{equilibria22}). 

As $c$ varies, there can exist $0$, $1$ or $2$ interior equilibria.
When $c<-0.1$, there are no interior equilibria. When $c=-0.1$ and $c=0.3$, there is a unique interior equilibrium $(0.5,0.5)$. If $-0.1<c<0$, there exist two  interior equilibria $(\frac{0.4\pm \sqrt{0.04+0.4c}}{2(0.3-c)}, \frac{0.4\pm \sqrt{0.04+0.4c}}{2(0.3-c)})$. And for $ 0\leq c <0.3$ and $  c \geq 0.3$,
 system (\ref{eesr2}) has a unique interior equilibrium $(\frac{0.4- \sqrt{0.04+0.4c}}{2(0.3-c)}, \frac{0.4- \sqrt{0.04+0.4c}}{2(0.3-c)})$.
 
 Note that although $\frac{0.4- \sqrt{0.04+0.4c}}{2(0.3-c)}$ is not defined at $c=0.3$, from L'Hôpital's rule, one has  $\lim _{c\rightarrow 0.3}\frac{0.4- \sqrt{0.04+0.4c}}{2(0.3-c)}=0.5$. Thus, one can summarize that
 there are two branches of interior equilibria: one branch, denoted by
 \begin{equation}\label{equilibriums1}
     \begin{aligned}
      (x_{s1}^*,r_{s1}^*):=\left(\frac{0.4+\sqrt{0.04+0.4c}}{2(0.3-c)}, \frac{0.4+ \sqrt{0.04+0.4c}}{2(0.3-c)}\right),
     \end{aligned}
 \end{equation} 
only exists for $c\in (-0.1,0)$; the other branch exists  for all $ c\geq -0.1 $, and is denoted by
 \begin{equation}\label{equilibriums2}
     \begin{aligned}
      (x_{s2}^*,r_{s2}^*):=\left(\frac{0.4-\sqrt{0.04+0.4c}}{2(0.3-c)}, \frac{0.4- \sqrt{0.04+0.4c}}{2(0.3-c)}\right).
     \end{aligned}
 \end{equation} 

Note that when $c<-0.1$, there exist no interior equilibria, and  the equilibrium $(1,1)$ is locally stable because it has two negative eigenvalues. The edge equilibrium $(0,0)$ is a saddle point whose stable manifold is actually the edge $\{(x,r)\in [0,1]^2:x=0\}$ because of the invariance. It can be checked that the vector fields in the close neighborhood of this edge all point inwards. It is impossible to have a homoclinic orbit between the unstable and stable manifold of $(0,0)$. Therefore, we have the following claim.
\begin{lem}\label{globalstable11}
When $c<-0.1$, the edge equilibrium $(1,1)$ is asymptotically stable in (\ref{eesr2}), and almost all trajectories starting from $[0,1]^2$  converge to it.
\end{lem}
This result can be proved straightforwardly by applying the Poincar\'e-Bendixson theorem \cite{Hirsch:04}. We omit the proof here due to the space limit. {\B This result implies that the full resource state can always be achieved   provided that the payoff difference $c$ between the two strategies in the extreme case, $r=1$, is sufficiently small.}

Now consider the Jacobian $J_s^*$ at $(x_s^*, r^*_s)$ which is given by
\begin{equation}\label{injacobian1}
 \begin{bmatrix}
 (x_s^*-{x_s^*}^2)((0.3-c)r_s^*+0.1)&(x_s^*-{x_s^*}^2)((0.3-c)x_s^*-0.5)\\
 \epsilon(0.6r_s^*+0.2)&-\epsilon(0.6r_s^*+0.2)
 \end{bmatrix}.   
\end{equation}
The determinant and trace are given as below:
\[\text{det}(J_s^*)=-\epsilon(0.6x_s^*+0.2)(x_s^*-{x_s^*}^2)(2(0.3-c)x_s^*-0.4),\]
\[\text{tr}(J_s^*)=(x_s^*-{x_s^*}^2)((0.3-c)x_s^*+0.1)-\epsilon(0.6x_s^*+0.2).\]
The eigenvalues can be calculated
\begin{equation}\label{lambda}
 \lambda_{1,2}(J_s^*)=\frac{\text{tr}(J_s^*)\pm \sqrt{{\text{tr}(J_s^*)}^2-4\text{det}(J_s^*)}}{2}.   
\end{equation}
It can be checked that the determinant of Jacobian at $(x_{s1}^*, r^*_{s1})$ is always negative, i.e., 
\[
 \text{det}(J_{s1}^*)=-\epsilon(0.6x_{s1}^*+0.2)(x_{s1}^*-{x_{s1}^*}^2)\sqrt{0.4(0.1+c)}<0
\] 
for $c\in (-0.1,0)$.
Then, from (\ref{lambda}), it can be inferred that the term beneath the square root sign is always positive, which implies the eigenvalues of $J_{s1}^*$ are real and have different signs. Therefore, the equilibrium $(x_{s1}^*, r^*_{s1})$ will always be unstable. 

In contrast, the determinant of Jacobian at $(x_{s2}^*, r^*_{s2})$ is non-negative because 
\[
 \text{det}(J_{s2}^*)=\epsilon(0.6x_{s2}^*+0.2)(x_{s2}^*-{x_{s2}^*}^2)\sqrt{0.4(0.1+c)}\geq 0
\]
for $c\geq -0.1$. And only when $c=-0.1$, the determinant equals $0$. When $\text{tr}(J_{s2}^*)=0$, the eigenvalues of $J_{s2}^*$ will be purely imaginary if the corresponding determinant is not equal to $0$. 
We let  $\text{tr}(J_{s2}^*)=0$, and it results in  
\begin{equation}\label{curve1}
\begin{aligned}
  \epsilon_c^s&:=\frac{(x_s^*-{x_s^*}^2)((0.3-c)x_s^*+0.1)}{0.6x_s^*+0.2}\\
  &=\frac{55c+30c\sqrt{10c+1}+\sqrt{10c+1}+50c^2-1}{(10c-3)(3\sqrt{10c+1}-10c+9)}.
\end{aligned}
\end{equation}
When $c$ varies, (\ref{curve1}) depicts a curve \begin{equation}
  T^s_0=\{(c,\epsilon): -0.1< c<0, \epsilon=\epsilon_c^s\}  
\end{equation}
in the parameter space $(c, \epsilon)$ as shown in Fig. \ref{fig:curve1} (a). On  $T^s_0$, the trace of Jacobian of $(x_{s2}^*,r_{s2}^*)$ is zero and the determinant is positive, thus it
has purely imaginary eigenvalues, i.e.,
\begin{equation}\label{pureimaginaryeigenvalue}
\begin{aligned}
\lambda_{1,2}(J_{s2}^*)|_{\epsilon_c^s}:=\pm i \omega_0
=\pm i \frac{(30c + 10c\hat{c} + \hat{c} - 1)\sqrt{2\hat{c}(\hat{c} + 3)}}{5(10c-3)^2},
\end{aligned}
\end{equation}
where $\hat{c}=\sqrt{10c+1}$.
\begin{figure}[htbp!]
    \centering
   	\subfloat[$T^s_0$]{\includegraphics[width=4cm]{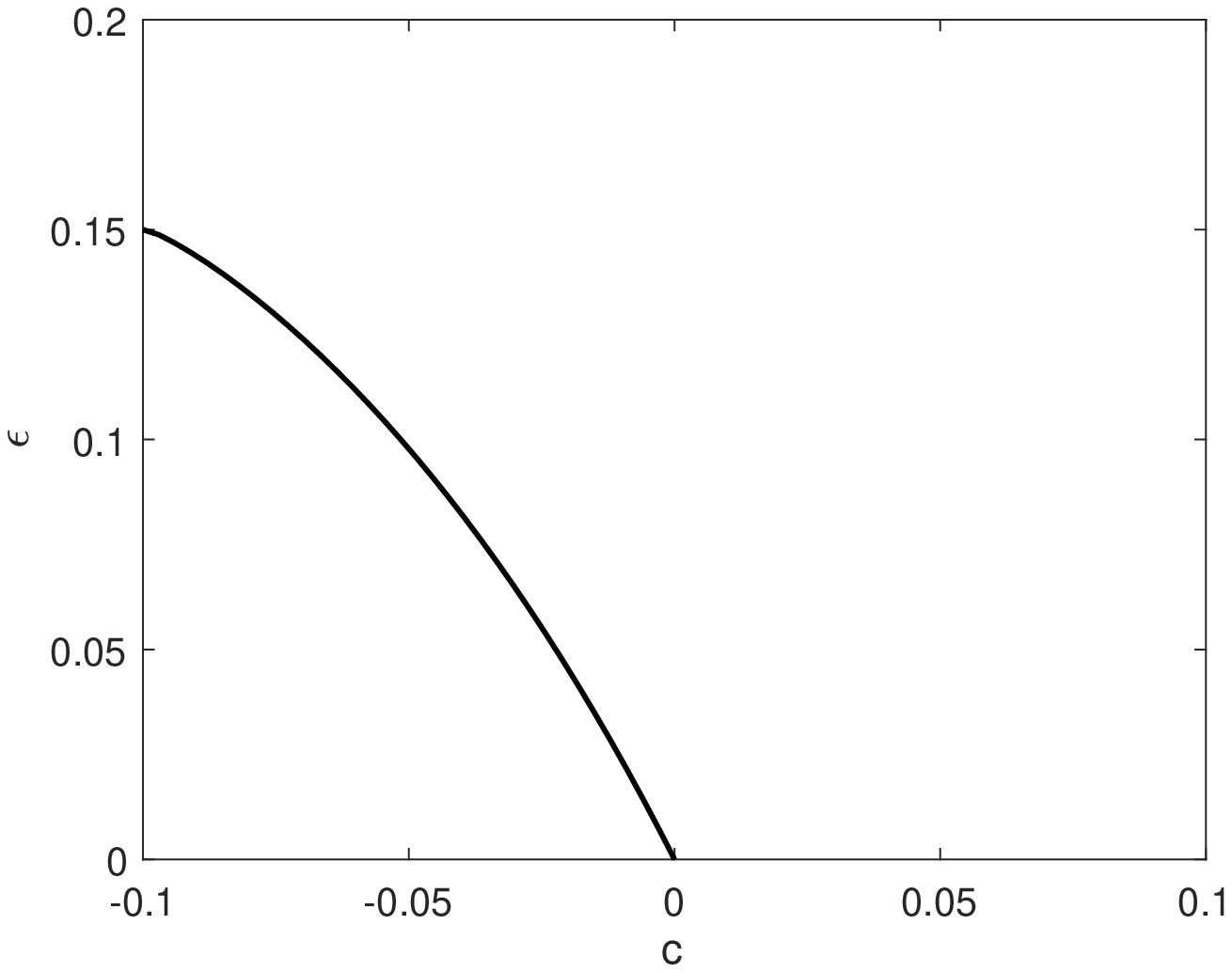}}~~
   	\subfloat[$T^e_0$]{\includegraphics[width=4cm]{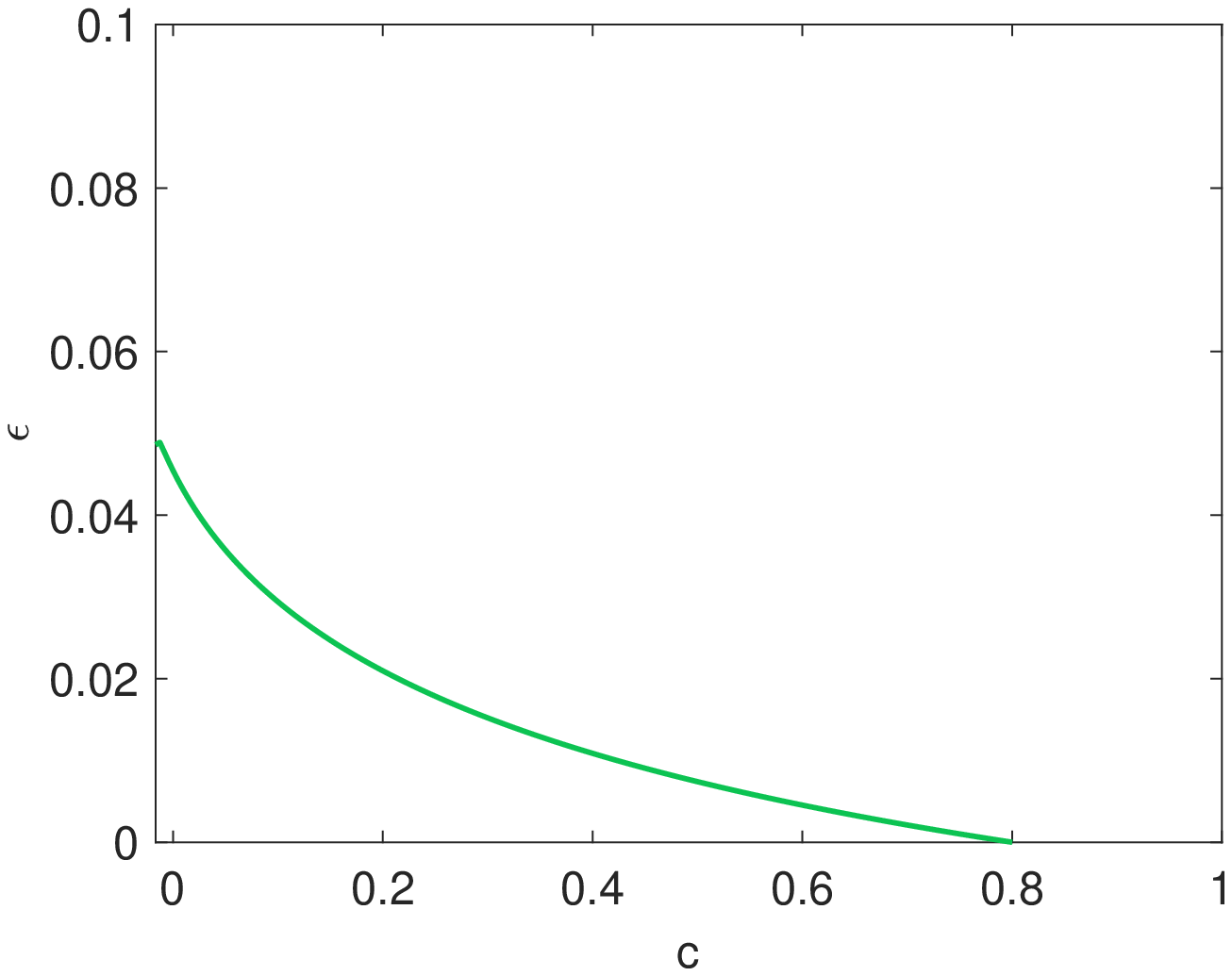}}
    \caption{Zero trace curves}
    \label{fig:curve1}
\end{figure}

Then, for sufficiently small neighborhood of $T_0^s$, the Jacobian (\ref{injacobian1}) has one pair of complex eigenvalues and can be denoted by
\begin{equation}\label{complexeigenvalues}
 \lambda_{1,2}^s(J_{s2}^*)=\mu(c,\epsilon)\pm i\omega (c,\epsilon) 
\end{equation}
such that $\mu(c,\epsilon)=0$ and $\omega(c,\epsilon)=\omega_0 >0$ for $c,\epsilon$ on $T_0^s$.

Now, we take $\epsilon$  as  the main parameter.
The following theorem states that the interior equilibrium $(x_{s2}^*,r_{s2}^*)$  undergoes a Hopf bifurcation.

\begin{theorem}\label{thm:hopfbifurcation1}
When $c\in (-0.1,0)$,
the interior equilibrium $(x_{s2}^*,r_{s2}^*)$ of system (\ref{eesr2}) undergoes a \emph{Hopf bifurcation} at $\epsilon_c^s$ defined by (\ref{curve1}).
\end{theorem}
\begin{proof}
From Hopf bifurcation theorem \cite[Theorem 1]{Guckenheimer:00},
a Hopf bifurcation occurs at the parameter point where the equilibrium  has  a  simple  pair  of  pure imaginary eigenvalues and the derivative of the real parts of the eigenvalues  is not zero. In (\ref{pureimaginaryeigenvalue}), one has already obtained that the eigenvalues are  purely imaginary  when $c\in (-0.1,0)$ and $\epsilon=\epsilon_c^s$.

Next, since the trace of the Jacobian is the sum of the eigenvalues, then in view of (\ref{injacobian1}) and (\ref{complexeigenvalues}), one can check that the derivative of the real part of the eigenvalues is
\[\left.\frac{d \mu(c,\epsilon)}{d \epsilon}\right|_{\epsilon_c^s} =\frac{10c-9-3\sqrt{1+10c}}{10(3-10c)}<  0.\] 
Thus, it is confirmed that the equilibrium  $(x_{s2}^*,r_{s2}^*)$ undergoes a Hopf bifurcation when the given condition is satisfied.
\end{proof}

A Hopf bifurcation can be \emph{supcritical} or \emph{subcritical}  depending on  the sign of the \emph{first Lyapunov coefficient} $\ell_1^s$ \cite{Guckenheimer:00}.
We use the  algorithm in \cite{Kuznetsov:04} (see page 103) to calculate $\ell_1^s$, whose expression is given by (\ref{firstlc_s}) in the Appendix. It can be checked that $\ell_1^s$ equals $0$ when $c\approx -0.0889$. Moreover,  as shown in Fig. \ref{fig:curvelc} (a), $\ell_1^s$ is positive when $c\in (-0.1, -0.0889)$, and $\ell_1^s$ is negative when $c\in (-0.0889,0)$.

If $\ell_1^s>0$, the occurred Hopf bifurcation is called subcritical, and an unstable limit cycle will be generated from $(x_{s2}^*,r_{s2}^*)$ for $\epsilon$ in the vicinity of $\epsilon_c$. In contrast, a stable limit cycle will bifurcate from $(x_{s2}^*,r_{s2}^*)$ if $\ell_1^s<0$ because of the supercritical Hopf bifurcation. Some examples illustrating the two different types of limit cycles can be found in Fig. \ref{fig:limitcyles}.

Moreover, we can clarify the stability of the interior equilibrium $(x_{s2}^*,r_{s2}^*)$ as follows:
when $(c,\epsilon)$ is above $T_s^0$, i.e., $c>-0.1$, $\epsilon>\epsilon_c^s$ , the equilibrium is stable because of two negative eigenvalues; it is unstable when $-0.1<c<0$ and $ \epsilon<\epsilon_c^s$ for the reason of two positive eigenvalues.

\subsubsection{Gerneralized Hopf bifurcation}
Now we will study the case of $\ell_1^s=0$, and we have to consider the \emph{generalized Hopf bifurcation}.
Note that  $c\approx -0.0889$ when $\ell_1^s=0$. 
In this case, we have 
\[
        \epsilon_c^s\approx0.1429,~
        \omega_0\approx 0.0715.
\]
Then, using the numerical tool Matcont, the \emph{second Lypunov coefficient} $\ell_2^s$ can be calculated
  $ \ell_2^s=4.8221\times 10^{-5}\neq0$.

To confirm the generalized Hopf bifurcation at parameter values $(-0.0889, 0.1429)$, the regularity condition has to be satisfied \cite{Kuznetsov:04}.  One can verify this condition by checking if the determinant of Jacobian matrix of the map $(c,\epsilon) \rightarrow (\mu  (c,\epsilon),\ell_1^s(c,\epsilon))$ near $(-0.0889, 0.1429)$ is nonzero, namely
\begin{equation}
    \text{det} \left.\begin{pmatrix}
    \frac{\partial \mu}{\partial c}&\frac{\partial \mu}{\partial \epsilon}\\
    \frac{\partial \ell_1^s}{\partial c}&\frac{\partial \ell_1^s}{\partial \epsilon}\\
    \end{pmatrix}\right|_{(-0.0889, 0.1429)} \neq 0.
\end{equation}

The computation of $\ell_2^s$ and $\text{det} (\cdot) $ is tedious and complicated, so the process is omitted here. Instead, we use the numeric bifurcation plot to show that the generalized Hopf bifurcation does take place, as shown in Fig. \ref{fig:limitcyles} (f). Intuitively, the generalized Hopf bifurcation implies that the system has two limit cycles of opposite stability properties for some $(c,\epsilon)$ near the point $(-0.0889, 0.1429)$.

{\B For a given payoff parameter $c$, one can see that  the time-scale difference has an important influence on the system behaviors. There is a critical value of $\epsilon$, such that the interior equilibrium's stability changes as $\epsilon$ crosses it. It results in the existence of stable or unstable limit cycles in the phase space of system (\ref{eesr2}) for the neighborhood of the bifurcation point depending on the signs of  the first Lyapunov coefficients. Moreover, for some certain $c$, two  limit cycles, unstable and stable respectively, could coexist, which makes the dynamic behaviors even more complicated. This behavior corresponds to an interesting dynamic bistable phenomenon as discovered in \cite{Yu}, \cite{bistability}.}

Next, let us study the eco-evolutionary system (\ref{eees2}). The procedure is quite similar.
\subsection{Analysis of system (\ref{eees2})}
For system (\ref{eees2}), denote the interior equilibria  by $(x_e^*, r^*_e)$.
Solving (\ref{equilibria42}) yields the $r$-coordinate of the equilibria. When $c< -1/60$, there are no solutions to (\ref{equilibria42}), and thus the interior equilibrium does not exist. When $-1/60\leq c<2/15 $ and $c>2/15 $, one has $r^*_e=\frac{0.6-\sqrt{0.04+2.4c}}{2(0.4-3c)}$. As $ c=2/15 $, (\ref{equilibria42}) has a single solution $r^*_e=1/3$. Similarly, we can denote the $r$-coordinate of the equilibrium uniformly by $r^*_e=\frac{0.6-\sqrt{0.04+2.4c}}{2(0.4-3c)}$, because of
$\lim _{c\rightarrow 2/15}\frac{0.6-\sqrt{0.04+2.4c}}{2(0.4-3c)}=1/3$.

Because $x_e^*=1.8r_e^*/(1.2+0.6r_e^*)$ in this case, we can obtain that there is a unique interior equilibrium
\begin{equation}\label{equilibriume1}
     \begin{aligned}
      (x_e^*, r^*_e):=\left(\frac{3\sqrt{60c+1}-9}{\sqrt{60c+1}+60c-11},\frac{0.6-\sqrt{0.04+2.4c}}{2(0.4-3c)}\right),
     \end{aligned}
 \end{equation} 
when $c\geq -1/60$. Otherwise, there are no interior equilibria, in which case obviously one has the analogous claim as in Lemma \ref{globalstable11}. {\B It means that the full resource state is also approachable   under the similar condition for  the payoff difference as in system (\ref{eesr2}),  which confirms the enhanced similarity between the two systems.}

Next, we evaluate the Jacobian $J_e^*$ of $(x_e^*, r^*_e)$, i.e.,
\begin{equation}\label{injacobian2}
 \begin{bmatrix}
 (x_e^*-{x_e^*}^2)((0.3-c)r_e^*+0.1)&(x_e^*-{x_e^*}^2)((0.3-c)x_e^*-0.5)\\
 \epsilon(1.2+0.6r_e^*)&-\epsilon(1.8-0.8x_e^*)
 \end{bmatrix}.   
\end{equation}
Following the same procedure in the analysis of system (\ref{eesr2}), we obtain the determinant and trace of  $J_e^*$ as below
\begin{equation*}
    \begin{aligned}
     &\text{det}(J_e^*)=3\epsilon(x_e^*-{x_e^*}^2)\frac{120cr_e^* - 15r_e^* + 20c{r_e^*}^2 - {r_e^*}^2 + 14}{50(r_e^* + 2)},\\   
     &\text{tr}(J_e^*)=(x_e^*-{x_e^*}^2)((0.3-c)r_e^*+0.1)-\epsilon(1.8-0.8x_e^*).
    \end{aligned}
\end{equation*}
The sign of $\text{det}(J_e^*)$ is determined solely by the numerator $120cr_e^* - 15r_e^* + 20c{r_e^*}^2 - {r_e^*}^2 + 14$. Substituting $r^*_e=\frac{0.6-\sqrt{0.04+2.4c}}{2(0.4-3c)}$ into it yields
\begin{equation}\label{numerator}
  \frac{33c' - 525cc' - 215c + 1800c^2c' + 1500c^2 + 17}{2(15c - 2)^2},  
\end{equation}
where  $c'=\sqrt{60c+1}$.
It can be checked that (\ref{numerator}) is always positive for $c\geq -1/60$, which implies that $\text{det}(J_e^*)$ is positive.
We set $\text{tr}(J_e^*)=0$, and then obtain the zero trace curve 
\begin{equation}
    T^e_0=\{(c,\epsilon): -1/60< c<4/5, \epsilon=\epsilon_c^e\},
\end{equation}
where \begin{equation}\label{curve2}
\begin{aligned}
  \epsilon_c^e&:=\frac{(x_e^*-{x_e^*}^2)((0.3-c)r_e^*+0.1)}{1.8-0.8x_e^*}\\
  &=\frac{37cc'-125c+4c'-165c^2c'+255c^2+900c^3+1}{(2-15c)(12cc'-330c-c'+1080c^2+23)}.
\end{aligned}
\end{equation}
 $T^e_0$ is shown in the parameter space $(c,\epsilon)$ as Fig. \ref{fig:curve1} (b).

Clearly, on $T^e_0$,  $(x_{s2}^*,r_{s2}^*)$ 
has purely imaginary eigenvalues, i.e.,
\begin{equation}\label{pureimaginaryeigenvalue2}
\begin{aligned}
\lambda_{1,2}(J_e^*)|_{\epsilon_c^e}&:=\pm i \nu_0,
\end{aligned}
\end{equation}
where \[\begin{aligned}
 \nu_0=&
\Big((72(300c + 5c' - 55)(60c - 10cc' + 3c' - 13)\cdot\\
&(99c' - 1575cc'- 645c + 5400c^2c' + 4500c^2 + 51)\cdot\\
&(15cc' - 75c + c' + 1)^2)/(25(15c - 2)\cdot\\
&(300c - 40)(3c' - 540c + 63)(60c + c' - 11)\cdot\\
&(60cc' - 630c - 11c' + 1800c^2 + 61)^2)\Big)^{1/2}.  \end{aligned}\]
Then, for sufficiently small neighborhood of $T_0^e$, the Jacobian (\ref{injacobian2}) has one pair of complex eigenvalues  denoted by
\begin{equation}\label{complexeigenvalues2}
 \lambda_{1,2}^e(J_{e}^*)=\delta(c,\epsilon)\pm i\nu (c,\epsilon) 
\end{equation}
such that $\delta(c,\epsilon)=0$ and $\nu(c,\epsilon)=\nu_0 >0$ for $c,\epsilon$ on $T_0^e$.
Now, we have the following result.
\begin{theorem}\label{thm:hopfbifurcation2}
When $c\in (-1/60,4/5)$,
the interior equilibrium $(x_e^*,r_e^*)$ of system (\ref{eees2}) undergoes a \emph{ supercritical Hopf bifurcation} at $\epsilon_c^e$ defined by (\ref{curve2}).
\end{theorem}
\begin{proof}
The proof of the occurrence of Hopf bifurcation is similar to that of Theorem \ref{thm:hopfbifurcation1}. 
From (\ref{pureimaginaryeigenvalue2}), one has already known that the eigenvalues are a pair of pure imaginary eigenvalues, i.e., $\pm i \nu_0$, when $c\in (-1/60,4/5)$ and $\epsilon=\epsilon_c^e$.

In view of (\ref{injacobian2}) and (\ref{complexeigenvalues2}), one obtains the derivative of the real part of the eigenvalues as below
\[\left.\frac{d \delta(c,\epsilon)}{d \epsilon}\right|_{\epsilon_c^e} =\frac{3(c' - 180c + 21)}{10(60c + c' - 11)}<  0.\] 
Thus,  the equilibrium  $(x_{s2}^*,r_{s2}^*)$ undergoes a Hopf bifurcation at  $\epsilon_c^e$ for $c\in (-1/60,4/5)$. 
Next, we evaluate the first Lyapunov coefficient at $(x_{s2}^*,r_{s2}^*)$ which is expressed in (\ref{firstlc_e}) in Appendix.  It can be checked that $\ell_1^e$ is always negative for  $c\in (-1/60,4/5)$ as shown in Fig. \ref{fig:curvelc} (b), which confirms that the corresponding Hopf bifurcation is supercritical.
\end{proof}

The same argument  can be used to clarify the stability of the interior equilibrium $(x_{e}^*,r_{e}^*)$, i.e.,
when $c>-1/60$, and $\epsilon>\epsilon_c^e$ , the equilibrium is stable; it is unstable when $-1/60<c<4/5$ and $ \epsilon<\epsilon_c^e$.

{\B In system (\ref{eees2}), one can observe that similar Hopf bifurcation and stable limit cycles as in system (\ref{eesr2}) can take place, which confirms the time-scale difference's important impact on the dynamic behaviors of the two systems. However, the existence of the double limit cycles is not possible due to the negative first Lyapunov coefficient, no matter how big the payoff parameter is. }

\subsection{Numeric Illustration}
Some simulations are given in Fig. \ref{fig:limitcyles} to show the potential behaviors of systems (\ref{eesr2}) and (\ref{eees2}) respectively. We also give in the same figure  the two-parameter bifurcation plot for  system (\ref{eesr2}).

\begin{figure}[htbp!]
    \centering
    \subfloat[Stable eq. $(1,1)$ in system (5)]{\includegraphics[width=4cm]{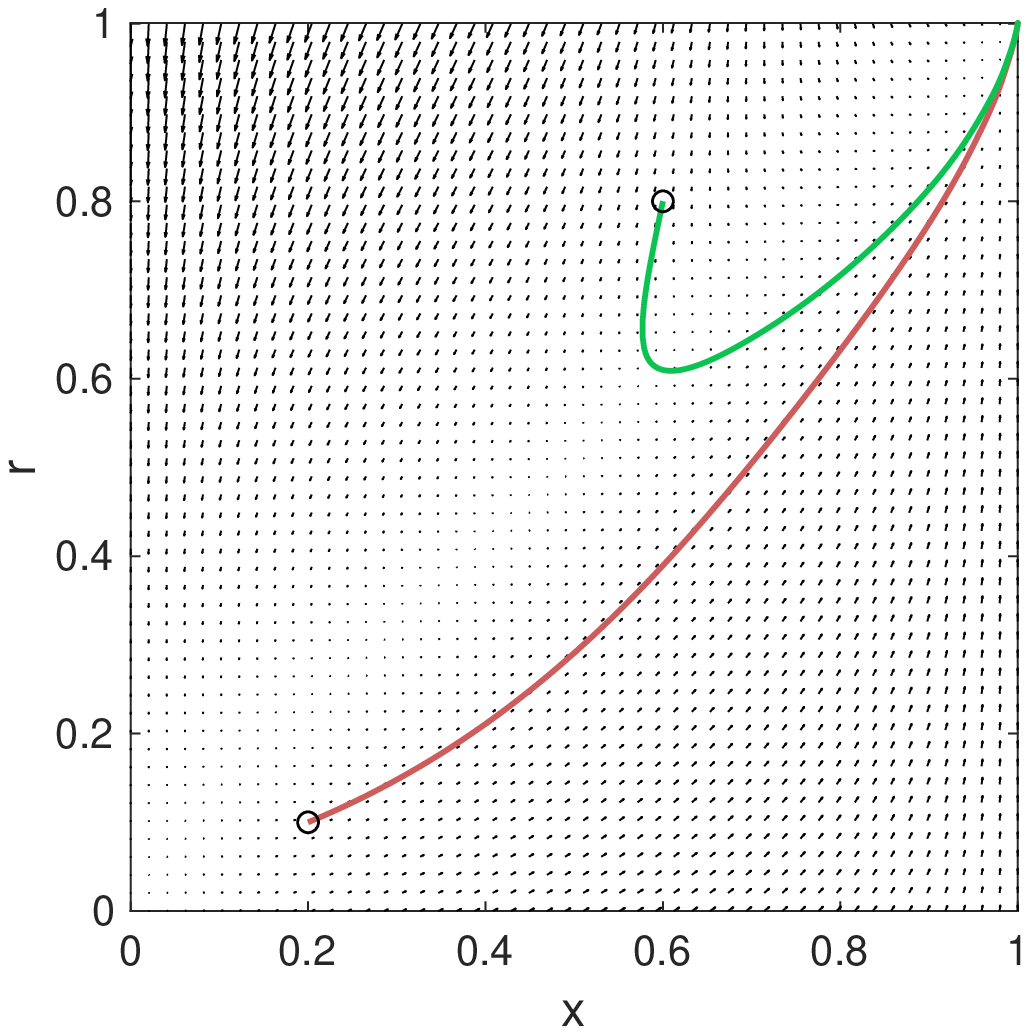}}~~
	\subfloat[Stable eq. $(1,1)$ in system (6)]{\includegraphics[width=4cm]{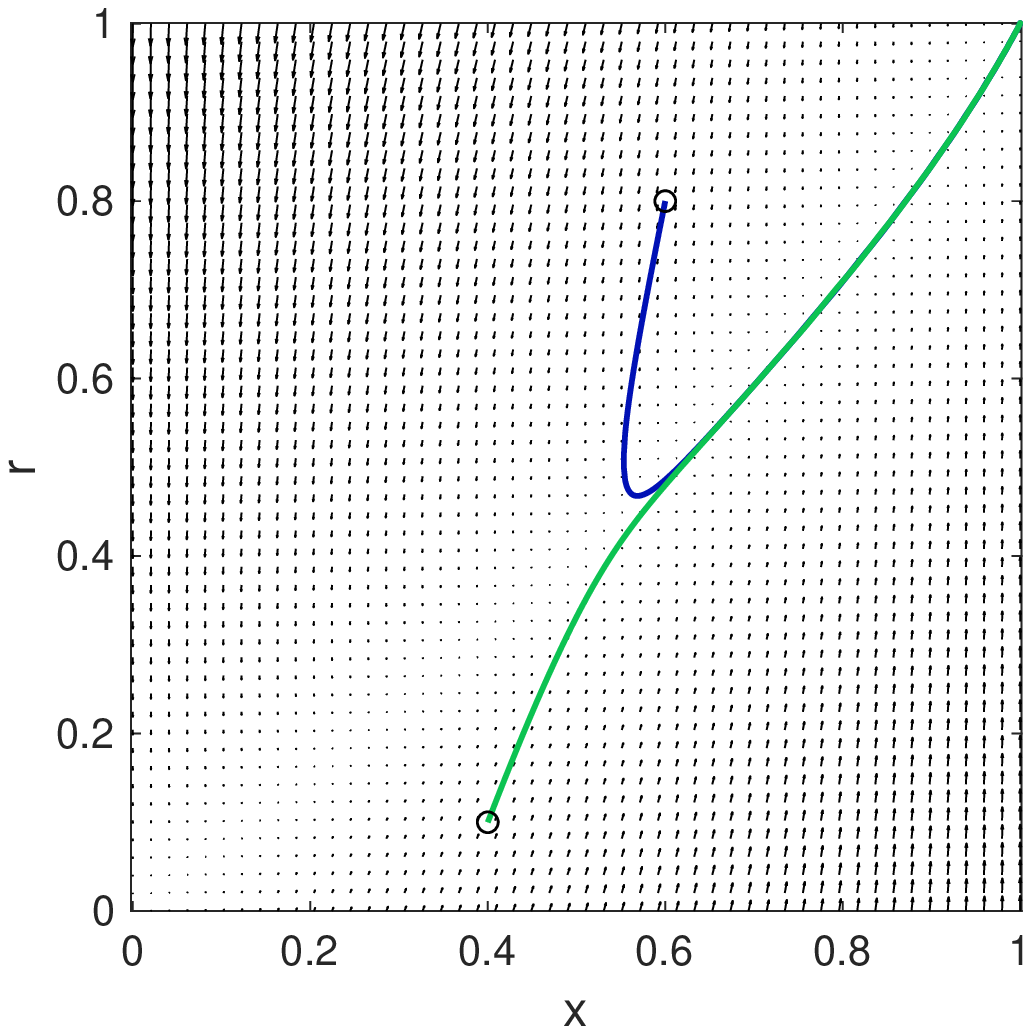}}\\
   	\subfloat[Stable limit cycle in system (5)]{\includegraphics[width=4cm]{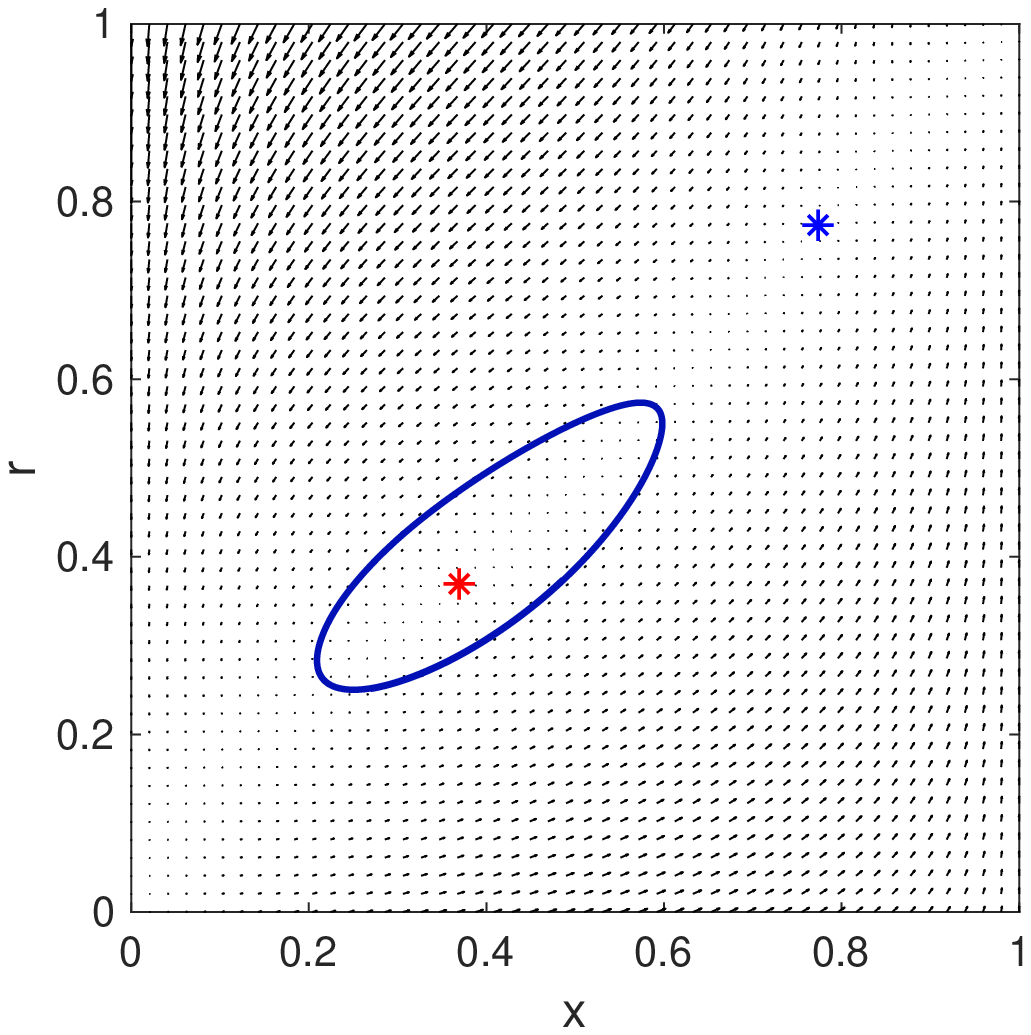}}~~
	\subfloat[Stable limit cycle in system (6)]{\includegraphics[width=4cm]{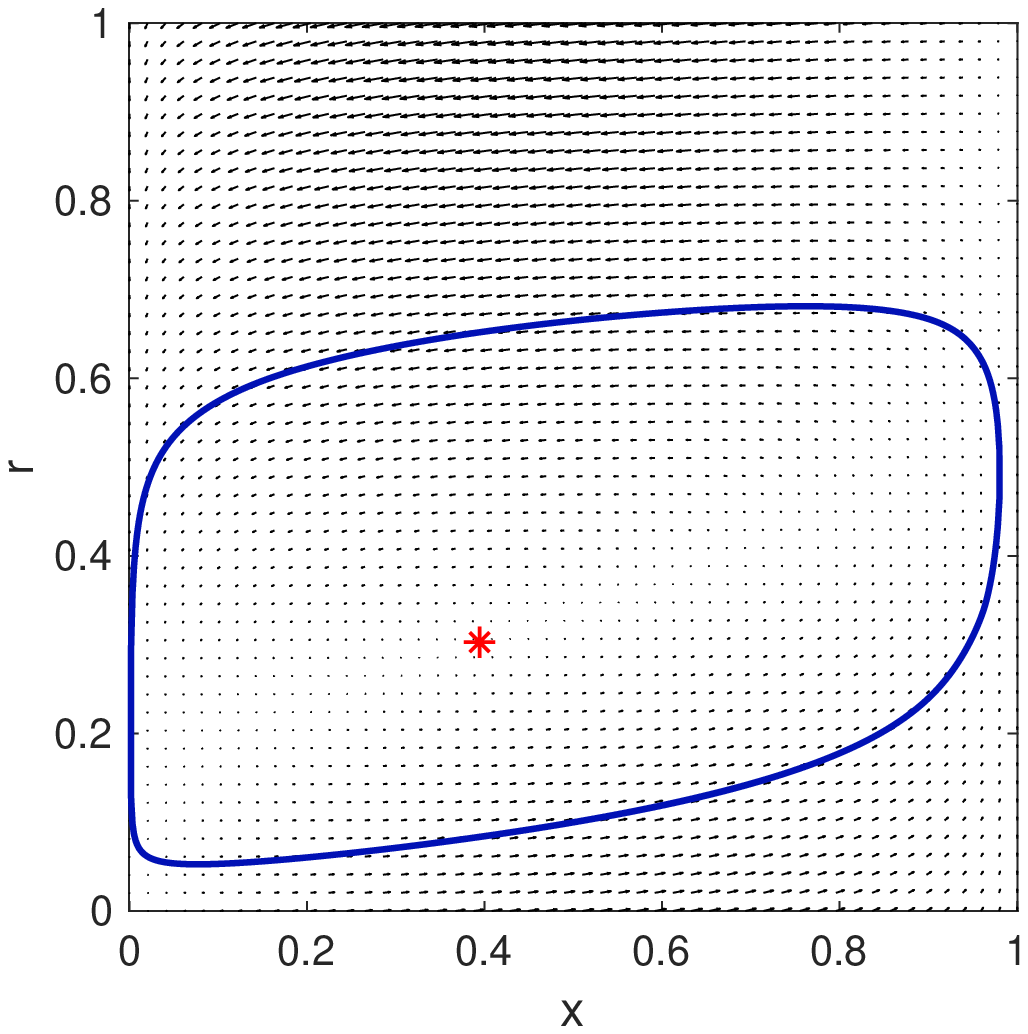}}\\
	\subfloat[Coexistence of two limit cycles in system (5)]{\includegraphics[width=4cm]{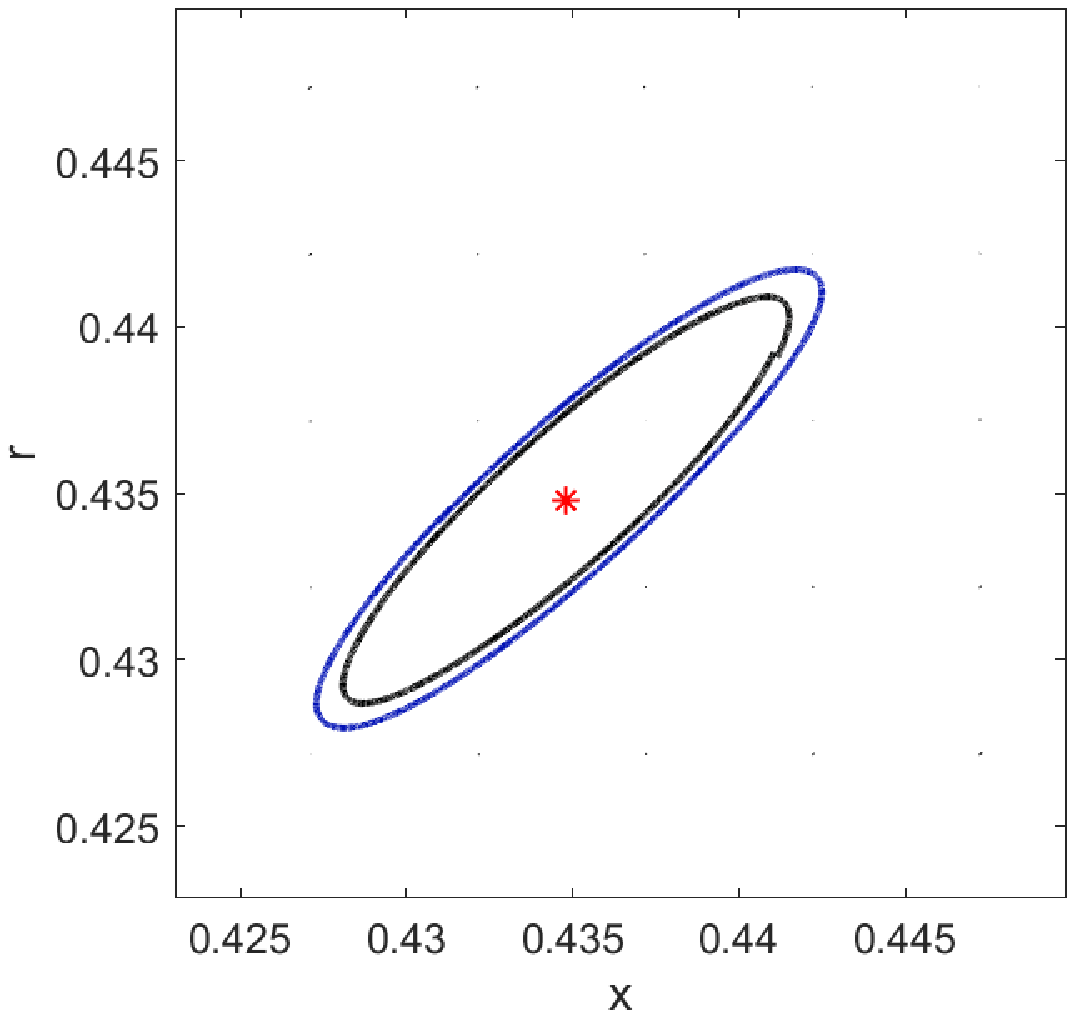}}~~
	\subfloat[Two-parameter bifurcation plot]{\includegraphics[width=4cm]{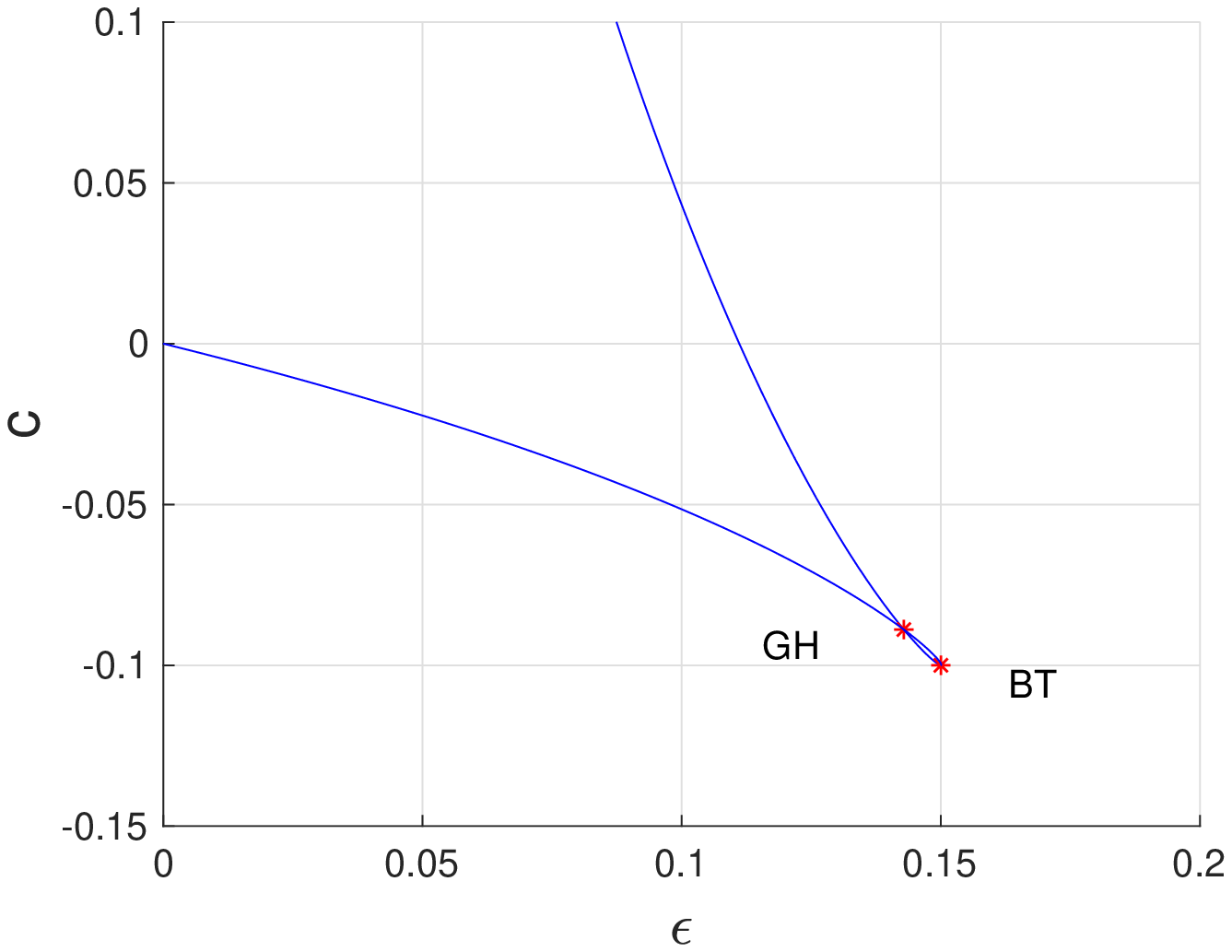}}
    \caption{The parameters in the plots are as follows: (a). $\epsilon=0.2$, $c=-0.15$; (b). $\epsilon=0.2$, $c=-0.05$; (c). $\epsilon=0.125$, $c=-0.05$; (d). $\epsilon=0.02$, $c=0.2$; (e). $\epsilon=0.145$, $c=-0.091$. In (f) GH stands for the generalized Hopf bifurcation, and BT stands for the Bogdanov-Takens bifurcation which is not covered in this work.}
    \label{fig:limitcyles}
\end{figure}

\section{Robust Dynamic Behaviors}
The bifurcation and stability analyses presented in last section are restricted to the given specific parameter conditions. Although the obtained results are rigorous and precise, it raises a question: if Hopf bifurcations and the resulting limit cycles can exist in the large space of parameters? In this section, we try to address this problem with suitable theoretical analysis as well as numeric simulations. 

\subsection{System (\ref{eesr2})}
Solving equation (\ref{equilibria2}) yields $r=\frac{(2b+d-a)\pm \sqrt{(a-d)^2+4bc}}{2(b-a+d-c)}$. Denote the potential interior equilibria of system (\ref{eesr2}) by $(\hat{x}_s, \hat{r}_s)_{\pm}=(\frac{(2b+d-a)\pm \sqrt{(a-d)^2+4bc}}{2(b-a+d-c)}, \frac{(2b+d-a)\pm \sqrt{(a-d)^2+4bc}}{2(b-a+d-c)})$.
Note that $\hat{x}_s= \hat{r}_s$, then  the Jacobian $\hat{J}_s$ of the equilibria  is given by
\begin{equation}\label{jacobian1s}
 \begin{bmatrix}
 (\hat{x}_s-{\hat{x}_s}^2)(\Delta\hat{x}_s+a-b)&(\hat{x}_s-{\hat{x}_s}^2)(\Delta\hat{x}_s-b-d)\\
 \epsilon(1-q(e_1\hat{x}_s+e_2(1-\hat{x}_s)))&-\epsilon(1-q(e_1\hat{x}_s+e_2(1-\hat{x}_s)))
 \end{bmatrix}    
\end{equation}
where $\Delta=b-a+d-c$.

The determinant and trace of $\hat{J}_s$ are given by
\begin{equation}
\begin{aligned}
  \text{det}(\hat{J}_s)=&-\epsilon(1-q(e_1\hat{x}_s+e_2(1-\hat{x}_s))(\hat{x}_s-{\hat{x}_s}^2)\cdot\\
  &~~(2\Delta\hat{x}_s+a-2b-d),   
\end{aligned}
\end{equation}
\begin{equation}
   \text{tr}(\hat{J}_s)=(\hat{x}_s-{\hat{x}_s}^2)(\Delta\hat{x}_s+a-b)-\epsilon(1-q(e_1\hat{x}_s+e_2(1-\hat{x}_s))).
\end{equation}

Recalling that the real interior equilibria should satisfy $\hat{x}_s\in (0,1)$, and $0< e_1<e_2<1/q$, then the sign of $\text{det}(\hat{J}_s)$ depends only on the term $(2\Delta\hat{x}_s+a-2b-d)$. If $\text{det}(\hat{J}_s)>0$, one has $(2\Delta\hat{x}_s+a-2b-d)<0$. It follows  $\pm \sqrt{(a-d)^2+4bc}<0$. It is impossible that $\sqrt{(a-d)^2+4bc}<0 $, which means the determinant of Jacobian at $(\hat{x}_s, \hat{r}_s)_{+}$ will be non-positive. Thus, we only need to focus on $(\hat{x}_s, \hat{r}_s)_{-}$ with $\sqrt{(a-d)^2+4bc}>0 $. 

Given $\text{det}(\hat{J}_s)>0$, the eigenvalues of $\hat{J}_s)$ will be both positive if $\text{tr}(\hat{J}_s)>0$ and negative if $\text{tr}(\hat{J}_s)<0$. The stability of $(\hat{x}_s, \hat{r}_s)_{-}$ switches when the value of $\text{tr}(\hat{J}_s)$ crosses $0$. Note that the position of interior equilibria is independent of parameter $\epsilon$ and $\text{tr}(\hat{J}_s)$ is linearly dependent on $\epsilon$. We can let $\text{tr}(\hat{J}_s)=0$, and it yields 
\begin{equation}\label{curve51}
\begin{aligned}
  \epsilon^s&:=\frac{(\hat{x}_s-{\hat{x}_s}^2)(\Delta\hat{x}_s+a-b)}{1-q(e_1\hat{x}_s+e_2(1-\hat{x}_s))},
\end{aligned}
\end{equation}
which is similar to (\ref{curve1}). By definition, it requires $\epsilon^s\in (0,1)$. However, due to too many parameters, it is still not clear to check this condition analytically. By fixing some parameters and varying the rest parameters, one can numerically  test if there exist such $\epsilon^s\in (0,1)$. We give some simulations in Fig. \ref{fig:robustlimitcyles} to visualize the shapes of $\epsilon^s$ in different coordinates, which shows that the existence of $\epsilon^s$ is not rare.

The change of stability of the equilibrium does not suffice to ensure the occurrence of Hopf bifurcation as stated in Theorem \ref{thm:hopfbifurcation1}. And it is challenging to calculate the first and second Lyapunov coefficients with so many unknown parameters. However, using Matcont, we can plot the bifurcation diagram as shown in Fig. \ref{fig:robustlimitcyles}. In this sense, we can say that the occurrence of (generalized) Hopf bifurcation is not rare in system (\ref{eesr2}) and the limit cycle behaviors can exist in different parameter conditions.  

\begin{figure}[htbp!]
    \centering
    \subfloat[$\epsilon^s$ shape in coordinates $(b, c, \epsilon)$]{\includegraphics[width=4cm]{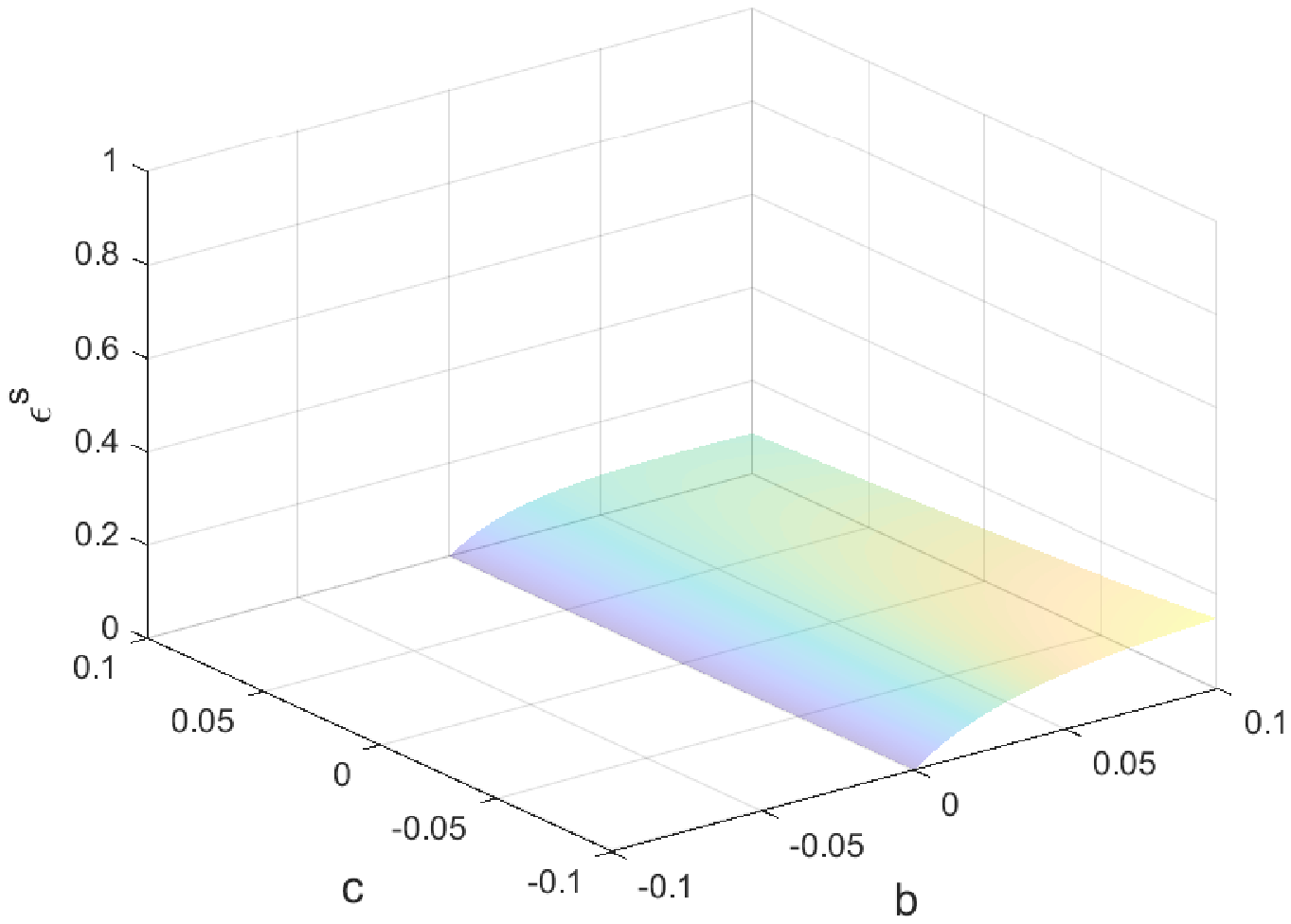}}~~
	\subfloat[$\epsilon^s$ shape in coordinates $(c, e_1, \epsilon)$ ]{\includegraphics[width=4cm]{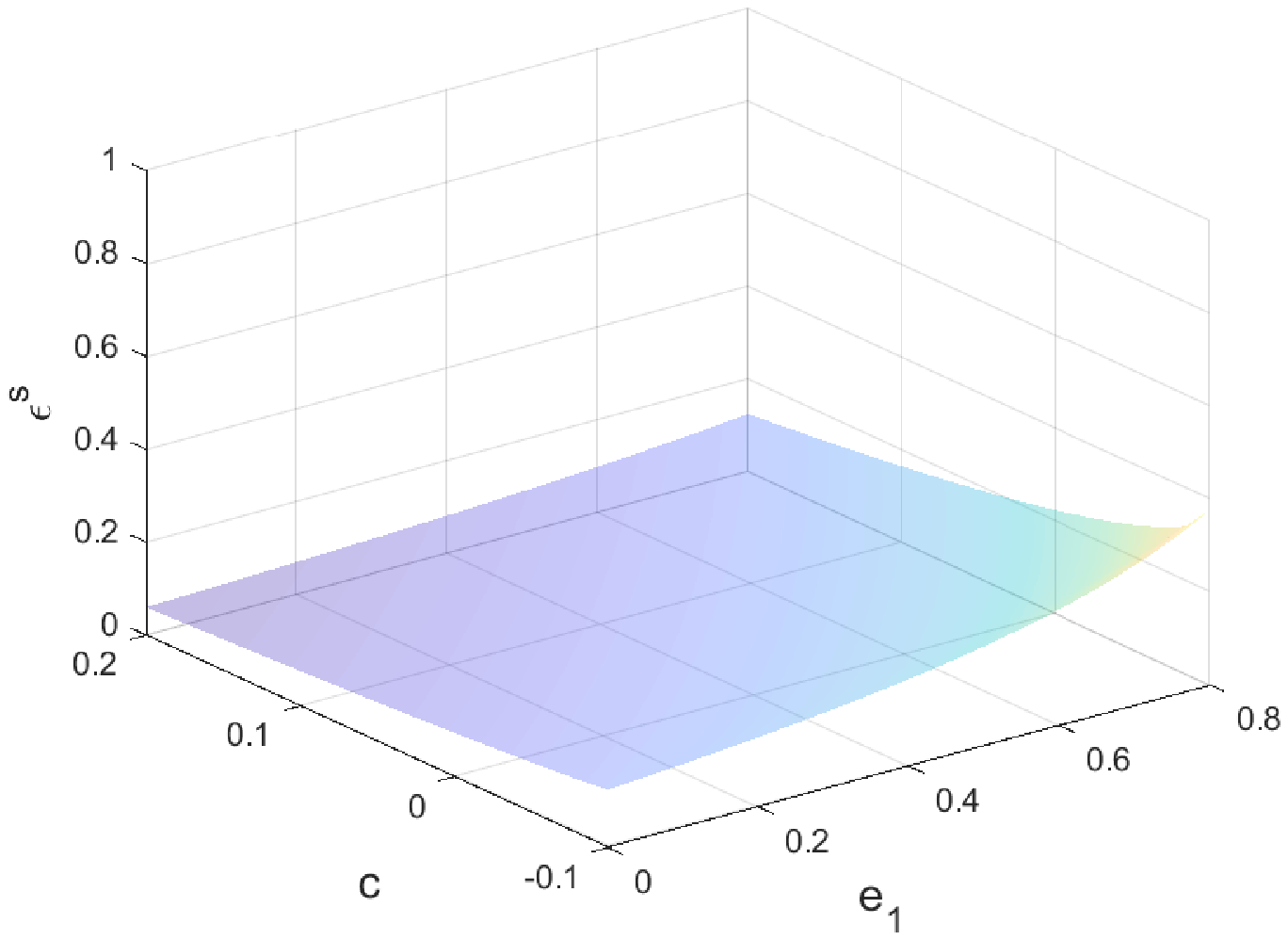}}\\
   	\subfloat[Two-parameter bifurcation plot 1]{\includegraphics[width=4cm]{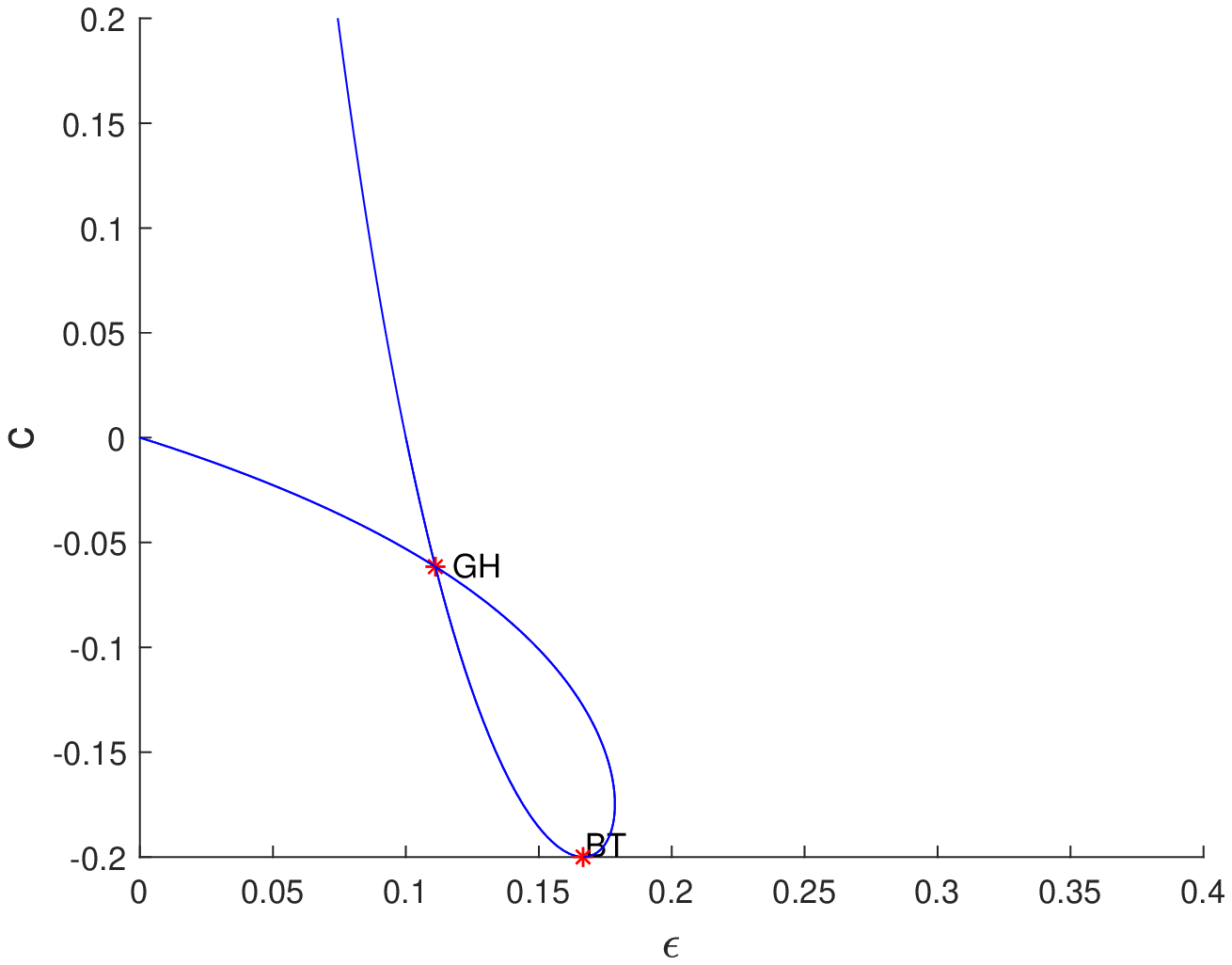}}~~
	\subfloat[Two-parameter bifurcation plot 2]{\includegraphics[width=4cm]{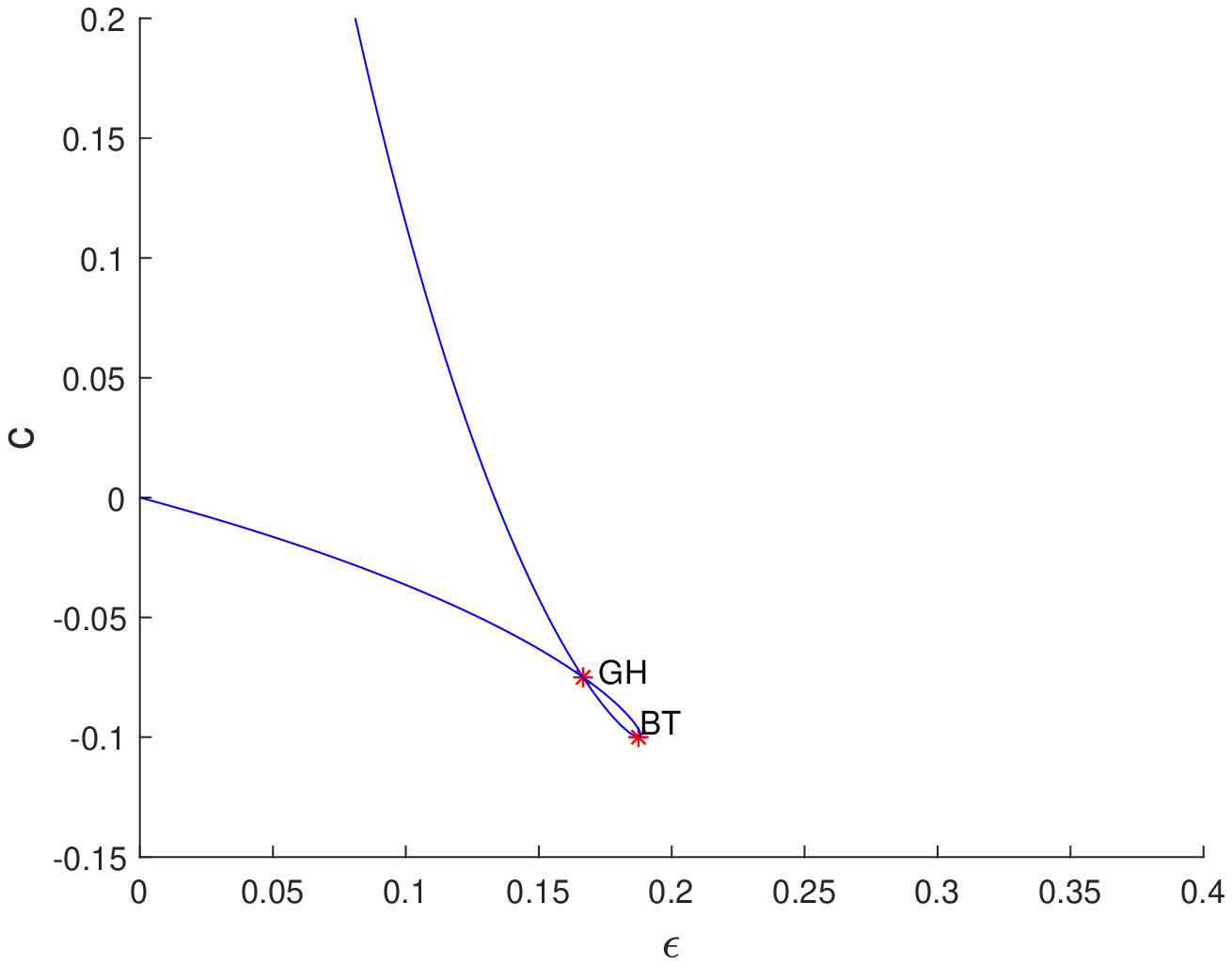}}
    \caption{The parameters in  (a) and  (b) are the same as in Section IV with one more free parameter respectively, i.e., $b$ and $e_1$; the parameters in  (c) and  (d) differ from previous setting of Fig. 3 (f) in $b=0.05$ and $e_1=0.4$ respectively.  }
    \label{fig:robustlimitcyles}
\end{figure}

\subsection{System (\ref{eees2})}
In a similar manner, we can analyze system (\ref{eees2}). Solving equation (\ref{equilibria3}) yields $r=\frac{q(ae_2-2be_1-de_1)+w(a-2b-d)\pm \sqrt{\alpha}}{2(\Delta(qe_2+w)-q(b+d)(e_2-e_1))}$ where $\alpha=w(2b - a + d) + q(2be_1 - ae_2 + de_1))^2 + b(w + e_1q)(4(w + e_2q)(a - b + c - d) - 4q(b + d)(e_1 - e_2))$. Note that $x=\frac{(qe_2+w)r}{(qe_1+w)+q(e_2-e_1)r}$, then the potential interior equilibria can be denoted by $(\hat{x}_e,\hat{r}_e)_\pm=(\frac{(qe_2+w)\hat{r}_e}{(qe_1+w)+q(e_2-e_1)\hat{r}_e},\frac{q(ae_2-2be_1-de_1)+w(a-2b-d)\pm \sqrt{\alpha}}{2(\Delta(qe_2+w)-q(b+d)(e_2-e_1))})$.

The  Jacobians $\hat{J}_e$ at the equilibria  are
\begin{equation}\label{jacobian1s}
 \begin{bmatrix}
 (\hat{x}_e-{\hat{x}_e}^2)(\Delta\hat{r}_e+a-b)&(\hat{x}_e-{\hat{x}_e}^2)(\Delta\hat{x}_e-b-d)\\
 \epsilon(qe_1+w+\hat{r}_eq(e_2-e_1))&-\epsilon(qe_2+w+\hat{x}_eq(e_1-e_2))
 \end{bmatrix}.    
\end{equation}
Clearly, one can obtain the determinant and trace of $\hat{J}_e$:
\begin{equation}
\begin{aligned}
  \text{det}(\hat{J}_e)=&-\epsilon(\hat{x}_e-{\hat{x}_e}^2)((w + e_2q - \frac{q\hat{r}_e(w + qe_2)(e_1 - e_2)}{w + qe_1 - q\hat{r}_e(e_1 - e_2)})\cdot\\
  &~(a - b + \Delta \hat{r}_e) - (b + d - \frac{\Delta \hat{r}_e(w + qe_2)}{w + qe_1 - q\hat{r}_e(e_1 - e_2)})\cdot\\
  &~(w + qe_1 + q\hat{r}_e(e_1 - e_2))),   
\end{aligned}
\end{equation}
\begin{equation}
\begin{aligned}
   \text{tr}(\hat{J}_e)=&(\hat{x}_e-{\hat{x}_e}^2)(\Delta\hat{r}_e+a-b)\\
   &-\epsilon(qe_2+w+\frac{q(e_1-e_2)(qe_2+w)\hat{r}_e}{(qe_1+w)+q(e_2-e_1)\hat{r}_e}).
\end{aligned}
\end{equation}
Compared with system (\ref{eesr2}),
the expression of $\text{det}(\hat{J}_e)$ and $ \text{tr}(\hat{J}_e)$ are much more complicated, and it is impossible to get any clues to evaluate their properties. Therefore, we let $b$, $c$,  $\epsilon$ and $e_1$ be unfixed and fix the other parameters as  same as in Section IV. In this way, we have 

\begin{equation}
\begin{aligned}
  \text{det}(\hat{J}_e)=&-\epsilon(\hat{x}_e-{\hat{x}_e}^2)(\frac{1.8(e_1+1)(0.2 - b + \Delta \hat{r}_e)}{(e_1+1)+(0.8-e_1)\hat{r}_e}\\
  & + \frac{1.8\Delta \hat{r}_e(1 + e_1 + \hat{r}_e(e_1 - 0.8))}{(e_1+1)+(0.8-e_1)\hat{r}_e}-b -0.4),   
\end{aligned}
\end{equation}
\begin{equation}
\begin{aligned}
   \text{tr}(\hat{J}_e)=(\hat{x}_e-{\hat{x}_e}^2)(\Delta\hat{r}_e+0.2-b)-\epsilon\frac{1.8(e_1+1)}{(e_1+1)+(0.8-e_1)\hat{r}_e}.
\end{aligned}
\end{equation}
One can check that $\text{det}(\hat{J}_e)>0$ for $(\hat{x}_e,\hat{r}_e)_-$. 
It follows that
\begin{equation}\label{curve51}
\begin{aligned}
  \epsilon^e&:=\frac{(\hat{x}_e-{\hat{x}_e}^2)(\Delta\hat{r}_e+0.2-b)(e_1+1+(0.8-e_1)\hat{r}_e)}{1.8(e_1+1)}.
\end{aligned}
\end{equation}
In the simulations in Fig. \ref{fig:robustlimitcyles2}, we can see the shape of $\epsilon^e$, which implies such $\epsilon^e$ exists for a large space of parameters. Again, using Matcont we can show that Hopf bifurcations can happen while generalized Hopf bifurcation does not appear in different parameter conditions.

\begin{figure}[htbp!]
    \centering
    \subfloat[$\epsilon^e$ shape in coordinate $(b, c, \epsilon)$ ]{\includegraphics[width=4cm]{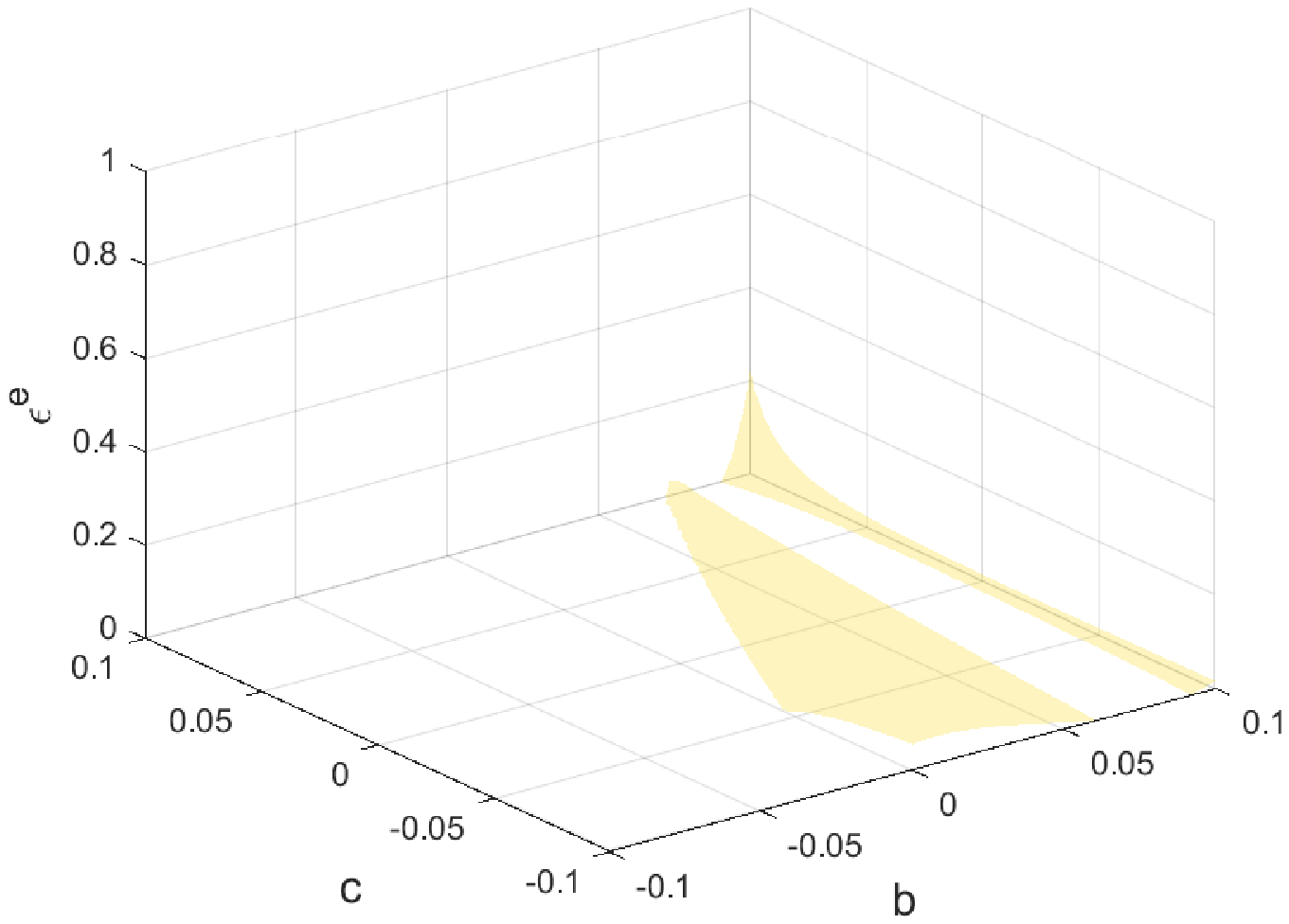}}~~
	\subfloat[$\epsilon^e$ shape in coordinate $(c, e_1, \epsilon)$ ]{\includegraphics[width=4cm]{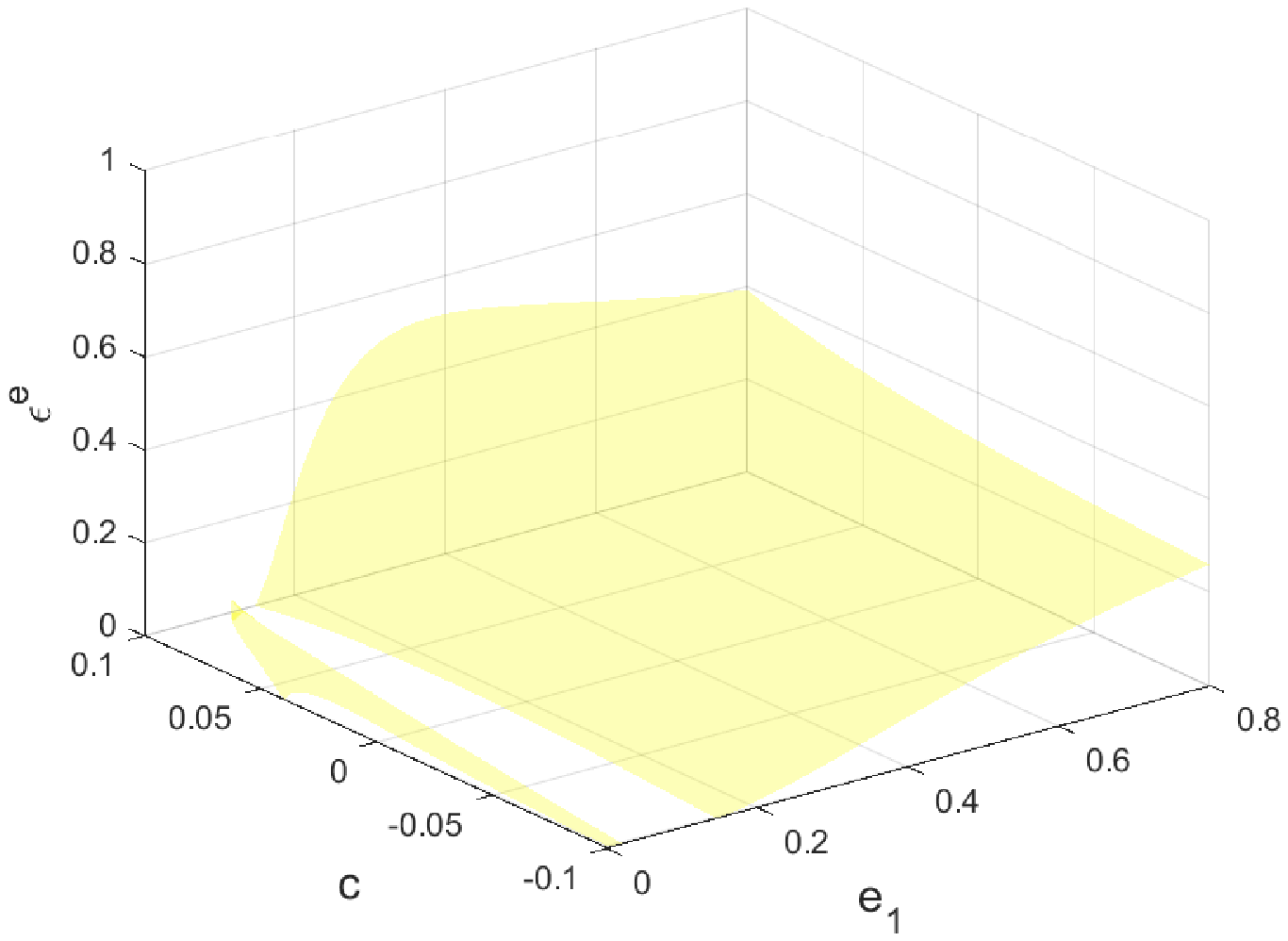}}\\
   	\subfloat[Two-parameter bifurcation plot 1]{\includegraphics[width=4cm]{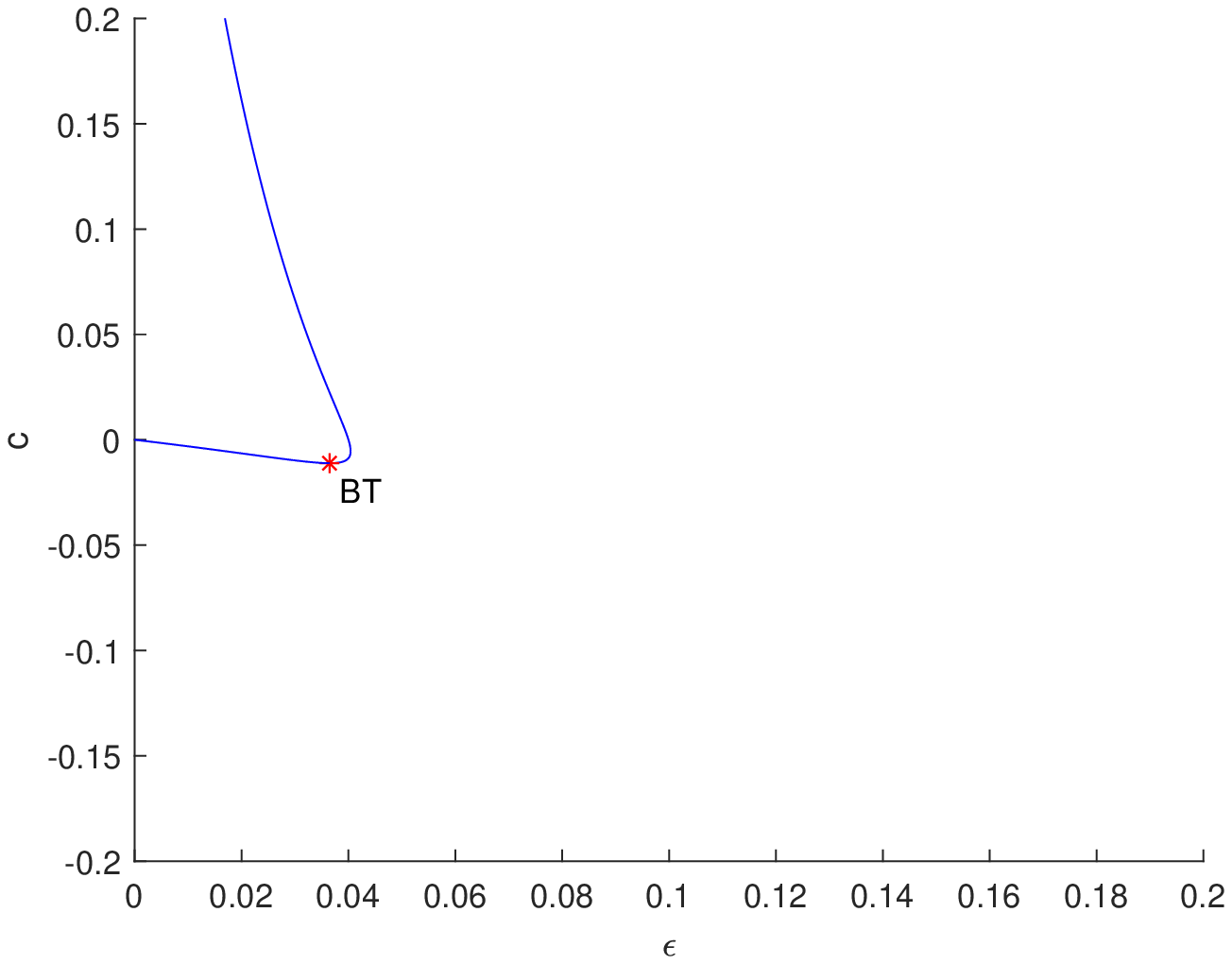}}~~
	\subfloat[Two-parameter bifurcation plot 2]{\includegraphics[width=4cm]{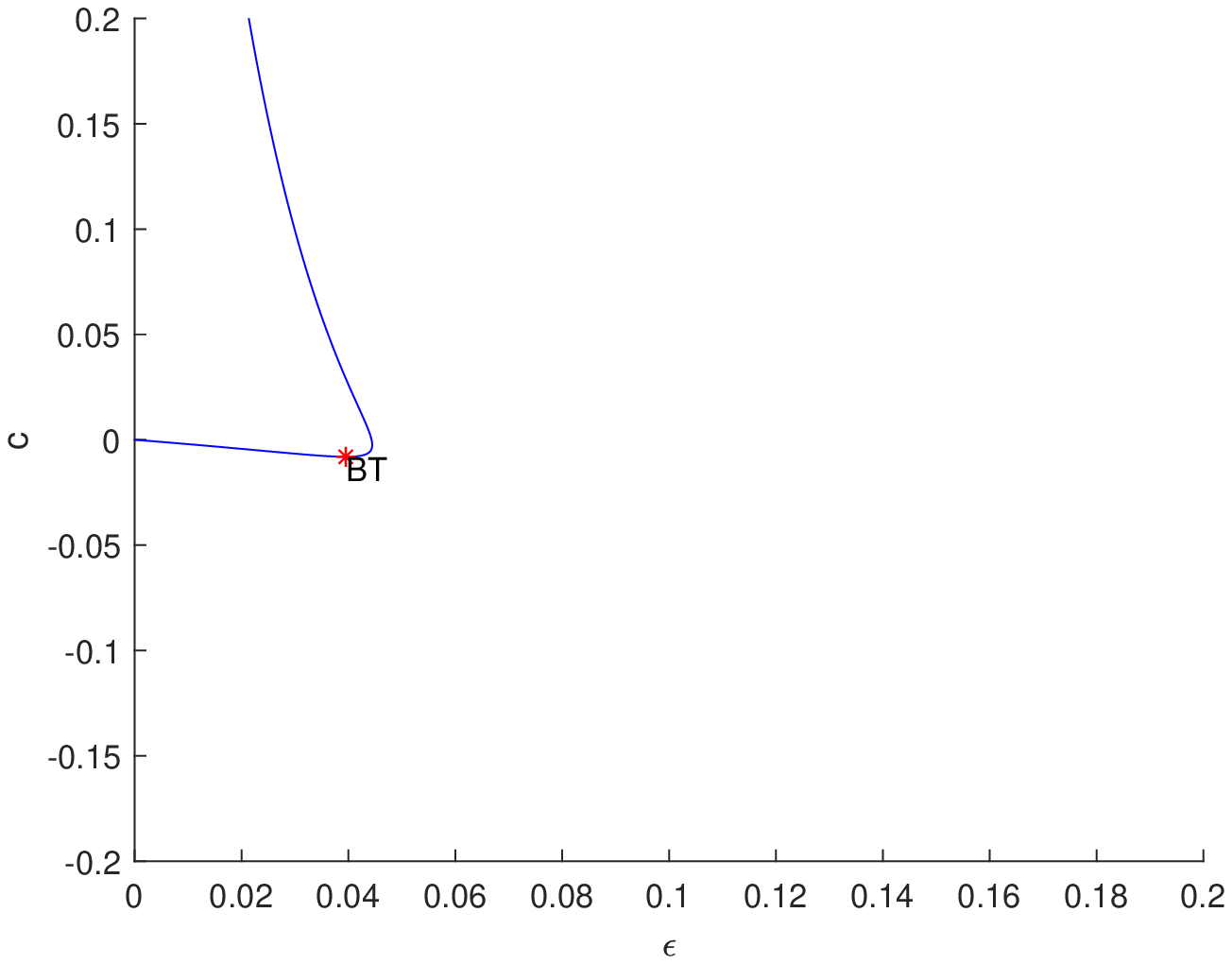}}
    \caption{The parameters' settings are the same as Fig. \ref{fig:robustlimitcyles}.  }
    \label{fig:robustlimitcyles2}
\end{figure}

To summarize, the above  analyses are not sufficient to ensure that Hopf bifurcations and limit cycles are robust in system (\ref{eesr2}) and (\ref{eees2}), and the numeric results are only valid for the chosen conditions. However, at least it gives a good indication for the robustness in some sense. This question remains to be unclear, unless more advanced mathematical tools are exploited further.
\section{CONCLUSION}
In this work, we have studied the eco-evolutionary dynamics under two different resource models.
Our main results have shown that the dynamic behaviors of the corresponding systems have similarities, but can be qualitatively different in the specialised conditions. They are able to exhibit limit cycles when the time-scale difference between the strategy and resource variables is sufficiently large. 
In particular, the system under the self-renewing resource feedback can display more complicated oscillating behaviors when the payoff parameters vary, e.g., two limit cycles of opposite stability properties. {\B However, these complicated
dynamic behaviors may be not possible in the system under the externally supplied resource feedback as far as our current analysis reveals. Moreover, preliminary analysis is conducted to investigate if the discovered dynamic behaviors persist in other parameter conditions.  
We show that different resource feedbacks dramatically alter integrated dynamics  and result in different and complicated dynamic behaviors.}

This work is a first step toward studying  eco-evolutionary dynamics as an analogue to classic Lotka–Volterra and MacArthur's consumer-resource models by considering dynamic-feedback games. There are a number of interesting future directions such as studying multiple populations \cite{Gong:18}, \cite{Kawano:19}, strategies  and  resources, {\B e.g.,   self-renewing and externally supplied resources exist at the same time. Other future directions  include carefully investigating the robustness of the discovered dynamics behaviors in the larger parameter space.}

\bibliographystyle{unsrt}
\bibliography{reference.bib} 

\begin{thebibliography}{10}

\bibitem{lotka1956elements}
A.J. Lotka.
\newblock Elements of physical biology (1925).
\newblock {\em Williams and Wilkins, Baltimore}, 1956.

\bibitem{volterra1926variazioni}
V.~Volterra.
\newblock Variazioni e fluttuazioni del numero d'individui in specie animali
  conviventi.
\newblock 1926.

\bibitem{Macarthur1967}
R.~MacAthur and R.~Levins.
\newblock The limiting similarity, convergence, and divergence of coexisting
  species.
\newblock {\em The American Naturalist}, 101(921):377--385, 1967.

\bibitem{CHESSON199026}
Peter Chesson.
\newblock Macarthur's consumer-resource model.
\newblock {\em Theoretical Population Biology}, 37(1):26--38, 1990.

\bibitem{Hofbauer:98}
J.~Hofbauer and K.~Sigmund.
\newblock {\em Evolutionary Games and Population Dynamics}.
\newblock Cambridge University Press, Cambridge, 1998.

\bibitem{Weitz2016}
J.~S. Weitz, C.~Eksin, K.~Paarporn, S.~P. Brown, and W.~C. Ratcliff.
\newblock An oscillating tragedy of the commons in replicator dynamics with
  game-environment feedback.
\newblock {\em Proceedings of the National Academy of Sciences},
  113(47):E7518--E7525, 2016.

\bibitem{Lin2019}
Y.~Lin and J.~S. Weitz.
\newblock Spatial interactions and oscillatory tragedies of the commons.
\newblock {\em Phys. Rev. Lett.}, 122:148102, 2019.

\bibitem{Tilman:20}
A.~R. Tilman, J.~Plotkin, and E.~Akcay.
\newblock Evolutionary games with environmental feedbacks.
\newblock {\em Nature Communications}, 11:915, 2020.

\bibitem{Cui2020}
W.~Cui, R.~Marsland, and P.~Mehta.
\newblock Effect of resource dynamics on species packing in diverse ecosystems.
\newblock {\em Phys. Rev. Lett.}, 125:048101, Jul 2020.

\bibitem{rand2017}
David~G Rand, Damon Tomlin, Adam Bear, Elliot~A Ludvig, and Jonathan~D Cohen.
\newblock Cyclical population dynamics of automatic versus controlled
  processing: An evolutionary pendulum.
\newblock {\em Psychological review}, 124(5):626, 2017.

\bibitem{bever1997}
James~D Bever, Kristi~M Westover, and Janis Antonovics.
\newblock Incorporating the soil community into plant population dynamics: the
  utility of the feedback approach.
\newblock {\em Journal of Ecology}, pages 561--573, 1997.

\bibitem{henrich2004}
Joseph Henrich.
\newblock Demography and cultural evolution: how adaptive cultural processes
  can produce maladaptive losses: the tasmanian case.
\newblock {\em American antiquity}, pages 197--214, 2004.

\bibitem{Gong:2020}
Lulu Gong, Weijia Yao, Jian Gao, and Ming Cao.
\newblock Limit cycles in replicator-mutator dynamics with game-environment
  feedback.
\newblock {\em IFAC-PapersOnLine}, 53(2):2850--2855, 2020.
\newblock 21th IFAC World Congress.

\bibitem{bistability}
Alexandra Erbach, Frithjof Lutscher, and Gunog Seo.
\newblock Bistability and limit cycles in generalist predator--prey dynamics.
\newblock {\em Ecological Complexity}, 14:48--55, 2013.

\bibitem{Yu}
P.~Yu and W.~Lin.
\newblock Complex dynamics in biological systems arising from multiple limit
  cycle bifurcation.
\newblock {\em Journal of Biological Dynamics}, 10(1):263--285, 2016.

\bibitem{hastings2018transient}
Alan Hastings, Karen~C Abbott, Kim Cuddington, Tessa Francis, Gabriel Gellner,
  Ying-Cheng Lai, Andrew Morozov, Sergei Petrovskii, Katherine Scranton, and
  Mary~Lou Zeeman.
\newblock Transient phenomena in ecology.
\newblock {\em Science}, 361(6406), 2018.

\bibitem{tikhonov2016}
Mikhail Tikhonov.
\newblock Community-level cohesion without cooperation.
\newblock {\em Elife}, 5:e15747, 2016.

\bibitem{Hirsch:04}
M.~W. Hirsch, S.~Smale, and R.~L. Devaney.
\newblock {\em Differential Equations, Dynamical Systems \& An Introduction to
  Chaos}.
\newblock Elsevier, Amsterdam, 2004.

\bibitem{Guckenheimer:00}
J.~Guckenheimer and P.~Holmes.
\newblock {\em Nonlinear Oscillations, Dynamical Systems and Bifurcations of
  Vector Fields}.
\newblock Springer, New York, 2000.

\bibitem{Kuznetsov:04}
Y.A. Kuznetsov.
\newblock {\em Elements of Applied Bifurcation Theory}.
\newblock Springer, New York, 2004.

\bibitem{Gong:18}
L.~Gong, J.~Gao, and M.~Cao.
\newblock Evolutionary game dynamics for two interacting populations in a
  co-evolving environment.
\newblock {\em Proceedings of the 57th IEEE Conference on Decision and
  Control}, pages 3535--3540, 2018.

\bibitem{Kawano:19}
Y.~Kawano, L.~Gong, B.D.O. Anderson, and M.~Cao.
\newblock Evolutionary dynamics of two communities under environmental
  feedback.
\newblock {\em IEEE Control Systems Letters}, 3(2):254--259, 2019.

\end{thebibliography}

\section{Appendix}
\begin{subequations}\label{firstl}
\begin{align}
\label{firstlc_s}
&\resizebox{1\hsize}{!}{$
\begin{aligned}
    \ell_1^s=&-5000000(c-0.3)^2\Big(((30c - 1)\hat{c} - 100(c + 0.1)^2)^{1/2}\cdot\\
    &~~(-50c^2 + 30\hat{c}c - 55c + \hat{c} + 1)^{1/2}(c + 0.1)\cdot\\
    &~~((-3/20c^3+ 9/8c^2 - 681/500c - 567/1250)\hat{c} + c^4 \\
    &~~- 13c^3/4 + 221c^2/200 + 1737c/500 + 117/250)\\
    &~~+ ((4500c^3 + 2300c^2 - 35c + 2)
    \hat{c} - 5000c^4 \\
    &~~- 15500c^3 - 1950c^2 - 25c + 2)^{1/2}\cdot\\
    &~~((-37500c^4 + 154000c^3 - 108125c^2 - 65130c - 5184)\hat{c}\\
    &~~ + 100000c^5 - 247500c^4 - 118500c^3 + 376225c^2\\
    &~~+ 90390c + 5184)/100000 \Big)\Big/\Big((10c + 3\hat{c} - 9)^2\cdot\\
    &~~(10c - 3\hat{c} + 1)^2
    (50^2 - 30\hat{c}c + 55c- \hat{c}- 1)\cdot\\
    &~~
    (100c^2 - 30\hat{c}c + 20c+ \hat{c} + 1)\Big).
    \end{aligned}$}\\
\label{firstlc_e}
&\resizebox{1\hsize}{!}{$
\begin{aligned}
    \ell_1^e=&151165440000000000000000000(c- 2/15)^4(c - 1/4)^6\cdot\\
    &~c(c - 4/5)^5c'\Big(3((3/10c^3 + 19/200c^2 - 1/1500c \\ &~+1/45000)c' + c^4 +
     119c^3/60 + 203c^2/1800 \\ &~-1/45000)c
    (60c + 7c' + 1)^{1/2}/20 + \sqrt{2}((11/40c^5 \\ &~-73/7200c^4 -
    989/216000c^3 + 29/120000c^2 \\
    &~+ 1/200000c - 
    1/3000000)c'+ (c + 1/60)\cdot\\
    &~(c^5 + 77/60c^4
    - 1043/3600c^3 
    + 43/2000c^2 - 9/10000c\\
    &~+ 1/50000))\Big)\Big/\Big((60c + 7c' + 1)^{3/2}(-11 - c' + 60c)^2\cdot\\
    &~(10cc' - 3c' + 60c - 13)^4(100c'c^2 + 6000c^3 - 185cc' \\
    &~-5000c^2 + 24c' + 715c - 24)^4\nu_0^3\Big).
    \end{aligned}$}
\end{align}
\end{subequations}

\begin{figure}[htbp!]
    \centering
   	\subfloat[$\ell_1^s$]{\includegraphics[width=4cm]{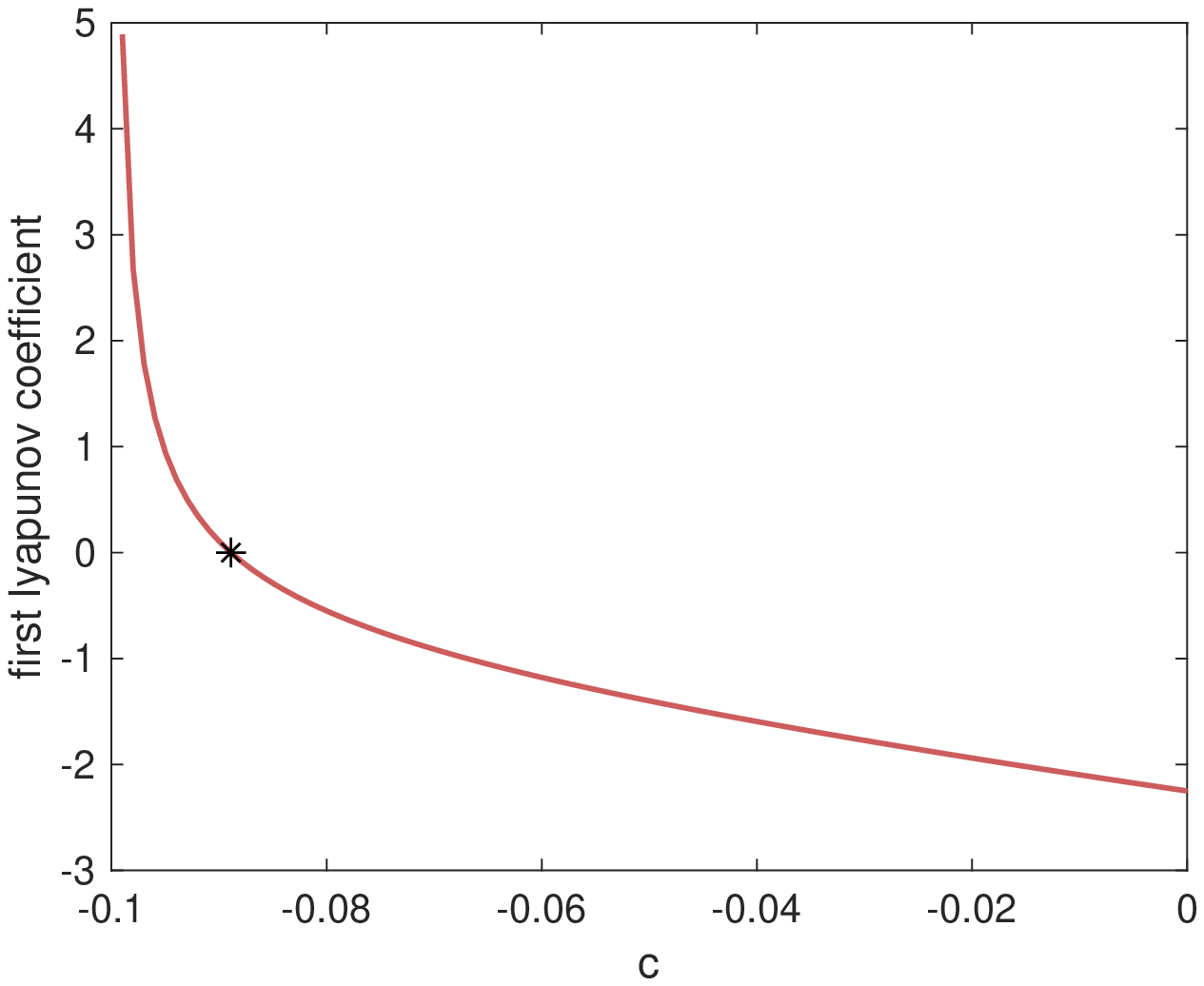}}~~
   	\subfloat[$\ell_1^e$]{\includegraphics[width=4cm]{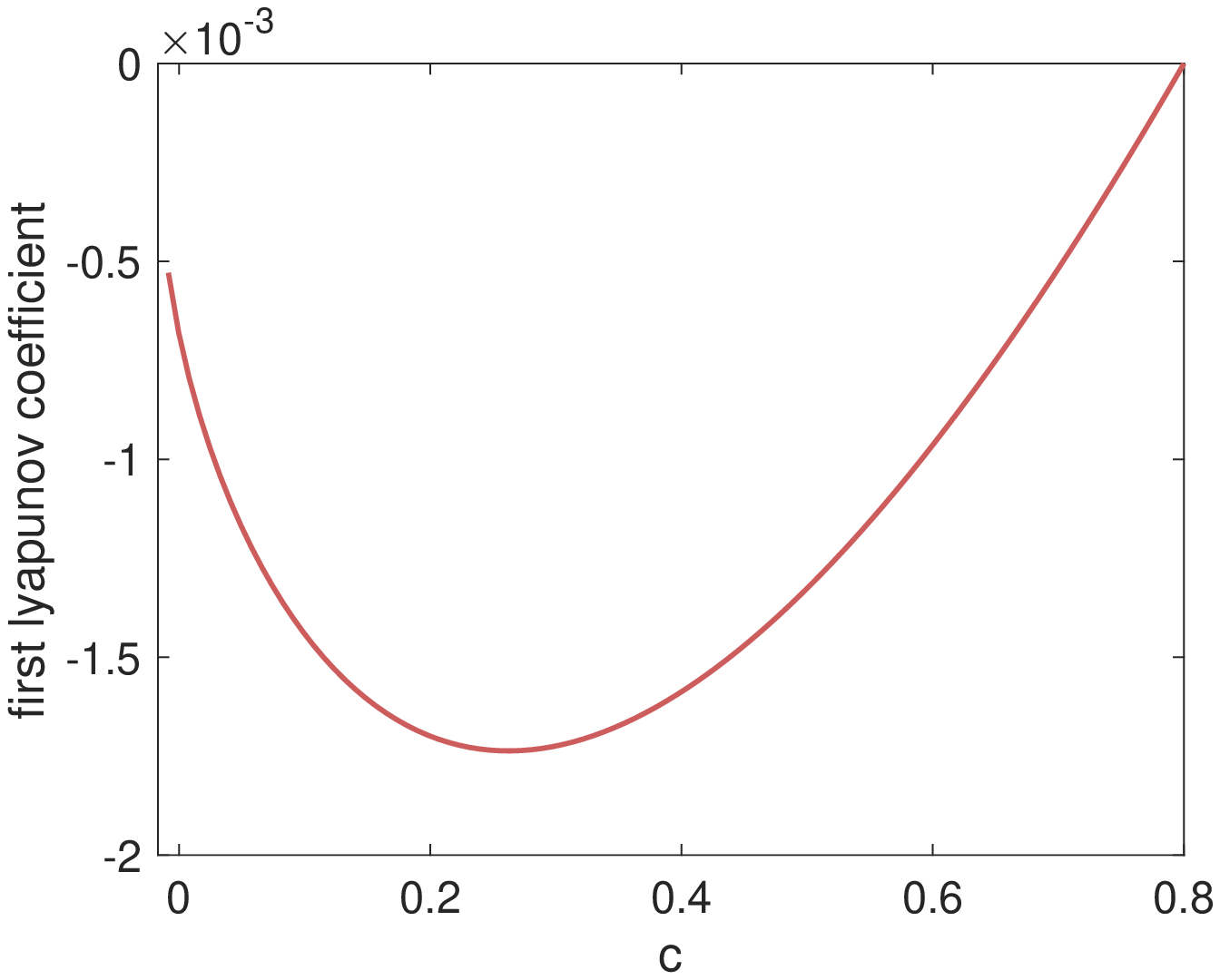}}
    \caption{The plots of the first Lyapunov coefficient. (a). the asterisk is the point where $\ell_1^s=0$. (b). $\ell_1^e<0$.}
    \label{fig:curvelc}
\end{figure}




\end{document}